\RequirePackage{fix-cm}
\documentclass{amsart}
\usepackage{latexsym}

\usepackage{amsmath,amssymb}
\usepackage{amsmath,amssymb, amsfonts}
\usepackage{verbatim}
\usepackage[dvips]{graphicx}

\newtheorem{proposition}{Proposition}
\newtheorem{theorem}{Theorem}
\newtheorem{definition}{Definition}
\newtheorem{corollary}{Corollary}
\newtheorem{lemma}{Lemma}
\newtheorem{remark}{Remark}

\numberwithin{equation}{section}

\usepackage{graphicx}
\usepackage{mathptmx}      
\usepackage{latexsym}
\newcommand{\Overline}{\underline{V}}
\newcommand{\Vdo}{\underline{V}}
\newcommand{\Vup}{\overline{V}}

\newcommand{\Me}{\mathcal{M}}
\newcommand{\Se}{\mathcal{S}}
\newcommand{\He}{\mathcal{H}}
\newcommand{\Swe}{\mathcal{S}^{\small\mathcal{W}}}
\newcommand{\Ref}{\mathcal{R}}
\begin{document}
\date{\today}
\title{Discrete, Non Probabilistic Market Models. \\
Arbitrage and Pricing Intervals.
}


\author{S.~E.~Ferrando}
\address{Department of Mathematics, Ryerson University, 350 Victoria St., Toronto,
Ontario M5B 2K3, Canada.} \email{ferrando@ryerson.ca}

\author{A.~L.~Gonzalez}
\address{Departamento de Matem\'atica. Facultad de Ciencias Exactas
y Naturales. Universidad Nacional de Mar del Plata. Funes 3350, Mar
del Plata 7600, Argentina.} \email{algonzal@mdp.edu.ar}

\author{I. L. Degano}
\address{Departamento de Matem\'atica. Facultad de Ciencias Exactas
y Naturales. Universidad Nacional de Mar del Plata. Funes 3350, Mar
del Plata 7600, Argentina.} \email{ivanldegano@yahoo.com.ar}

\author{M. Rahsepar}
\address{Department of Mathematics, Ryerson University, 350 Victoria St., Toronto,
Ontario M5B 2K3, Canada.} \email{a.rasehpar@hotmail.com}

\thanks{
The research of S.E. Ferrando is supported in part by an  NSERC
grant.}






\date{Received: date / Accepted: date}

\begin{abstract}
The paper develops general, discrete, non-probabilistic market models and
minmax price bounds leading to price intervals for  European options. The approach
provides the trajectory based analogue of martingale-like properties as well as a
generalization that allows a limited notion of arbitrage in the market
while still providing coherent option prices.
Several properties of the price bounds are obtained, in particular a connection with risk neutral pricing is established for trajectory markets associated to a continuous-time martingale model.
\end{abstract}
\maketitle


\section{Introduction}

The market model introduced by Britten-Jones and Neuberger  (BJ\&N) in \cite{BJN} incorporates several important market features: it reflects the discrete nature of financial transactions, it models the market in terms of observable trajectories and incorporates practical constraints such  as jump  sizes as well as methodological constraints in terms of the quadratic variation. Market frictions can also be included naturally. The book treatment in \cite{rebonato} also emphasizes the fundamental characteristics of the model's assumptions.
Our original interest in this approach stemmed from recent developments on non probabilistic market models
(\cite{AFO2011}, \cite{alvarez}), the  setting of \cite{BJN}  may be seen  as a natural discrete version
of these continuous-time models.
The present paper develops a framework that generalizes and formalizes
the original  BJ\&N model and, along the way uncovers some new phenomena not noticed in \cite{BJN}.

The framework of the paper is a discrete market model $\mathcal{M}= \mathcal{S} \times \mathcal{H}$,  $S = \{S_i\} \in \mathcal{S}$ is a sequence of real numbers
and $H = \{H_i\} \in \mathcal{H}$ a sequence of functions acting on $\mathcal{S}$ representing the portfolio holdings $H(S)= \{H_i(S)\}$ along $S$. The set of trajectories $\mathcal{S}$  plays a central stage in the developments; trajectories, as opposed to probabilities, are a basic observable phenomena, therefore, it is relevant to pursue developments based only on such characteristics of the market. The models are
{\it discrete} in the sense that we index potential portfolio rebalances, $H_i(S) \rightarrow H_{i+1}(S)$, by integer numbers. Otherwise, stock charts and investment amounts can take values in general subsets of the real numbers, data could flow in a time continuous manner and
portfolio rebalances could be triggered by arbitrary events without the  need to be associated to a time variable.

For a given European option $Z$,  we prove existence
of a pricing interval
$[\underline{V}(Z), \overline{V}(Z)]$ (see Definition \ref{priceBounds}) under conditions that allow for arbitrage opportunities in the market.
The co-existence of arbitrage and the price interval follows as a consequence of a worst case point of view and reflects a basic financial situation. Market players involved in the option's transaction  may need/prefer the option's contract sure benefits against
the potential arbitrage rewards. For  market agents transacting in the option, any market price
falling outside the proposed price interval generates an extended arbitrage opportunity (i.e. trading with the option is required) for one of the agents involved in the option's transaction. This extended arbitrage provides a profit for all elements of $\mathcal{S}$ and, so, it is riskless.

Part of the practical relevance of the interval $[\underline{V}(Z),
\overline{V}(Z)]$  depends on the relative sizes of the sets
$\mathcal{S}$ and $\mathcal{H}$, the collection of all trajectories and all portfolios, respectively,
occurring in $\mathcal{M}$. On the one hand, we should design
$\mathcal{S}$ to be large enough so that it allows for arbitrarily
close approximations of stock charts but not any larger so as not to
artificially enlarge the bounding interval. On the other hand,
$\mathcal{H}$ should include only portfolios that can be implemented
in practice (albeit in an idealized way) as the introduction of more
powerful, but impractical, hedging strategies may artificially
shrink the bounding interval. The fact that a minimization is required over the set of portfolios
directs attention to the issue of membership to $\mathcal{H}$, it is well known that judicious
choice of portfolio sets can change substantially the properties of the associated market in continuous-time
(see for example results on non-semimartingale processes in \cite{cheridito}). We also present an instance of this phenomenon in Section \ref{sec:arbitrageFreeMarkets}.

We ask: what are the fundamental path properties, independent of the probability measure, of a discrete time martingale, that permits to obtain no arbitrage results?
The simple notion of {\it arbitrage-free node}, contained in Definition
\ref{localDefinitions}, allows for probability free developments
of arbitrage-free markets.
The no arbitrage conditions obtained in our paper, see Corollaries
\ref{upDownCorollary} and \ref{secondResultOnNoArbitrage},
should be contrasted with the analogous
conditions in \cite{jarrow} (see also \cite{bender2}) in stochastic settings.
Related no arbitrage results in terms of properties of paths are in \cite{bender}.

As indicated,  with a worst case point of view, we uncover the following phenomena: there exists a rational price interval for a given option that does not introduce a relative arbitrage (in the sense of \cite{fernholz})
even though there may be arbitrage opportunities in the market. In our setting, this is reflected on the fact that
the set of portfolios $\mathcal{H}$
allows to define  the notion of $0$-neutral market (introduced originally in \cite{BJN}, but, in that reference, associated with no arbitrage):
\begin{equation} \nonumber
\inf_{H \in \mathcal{H}} \left\lbrace \sup_{S \in \mathcal{S}} \left\lbrace -\sum_{i=0}^{N_H(S)-1}~H_i(S)~\Delta_iS \right\rbrace \right\rbrace=0,
\end{equation}
(for details see Definition \ref{0NeutralDefinition}). It turns out that this notion is  a weakening of the no arbitrage property that still allows for a price interval (see Theorem \ref{havingAnIntervalTheorem})
and many martingale-like properties to go through in a trajectory based setting
while permitting a limited notion of arbitrage. One should compare $0$-neutrality
with the normalization $\rho(0)=0$ for a convex measure of risk $\rho$ (\cite{follmer2}).

$0$-neutral markets are closely related to trajectory  sets
obeying the local $0$-neutral property (introduced in Definition \ref{localDefinitions});
this latter condition should be contrasted with the notion of sticky processes (\cite{guasoni},
\cite{bender3}) which is fundamental to guarantee the removal of any possible arbitrage
in a model with non-zero transaction costs. Reference \cite{ferrando} obtains a similar result
for trajectory sets obeying the local $0$-neutral property under the presence of transaction costs.

To obtain $0$-neutral markets under the assumption of local $0$-neutrality of $\mathcal{S}$, one notices the existence of contrarian trajectories, these are elements of $\mathcal{S}$ that move in a contrarian manner to a given investment $H$ in such a way that makes the potential profits arbitrarily small (or negative). Under natural financial conditions, it also follows that the market player stops or liquidates her portfolio. These results are
developed in Section \ref{sec:localCharacterizations}.


A trajectory set is implicit in a stochastic process model; making trajectory sets a central object of interest is of relevance, in particular, when
there is insufficient information to assign a probability distribution  with confidence. An example is given by  the modelling of crashes in \cite{desmettre} where, the number, timing and size of a downwards stock change (a {\it crash}) is treated without probabilistic assumptions.
More importantly, giving trajectory sets a primary role changes the usual paradigm to model financial situations. Stochastic modeling relies on stochastic processes and the main input for their construction is a probability distribution; by contrast, the properties of their paths result as a by-product.
References \cite{AFO2011} and  \cite{alvarez} present continuous-time examples of trajectory sets which do not correspond to semimartingales. In the present paper we describe a general discrete example of a class of trajectory sets extending substantially the model in \cite{BJN}, in particular, the example incorporates trajectory dependent volatility. A computational and more detailed analysis is developed in the companion paper \cite{degano}. Section \ref{relationToMartingales} also introduces trajectory sets associated to continuous-time martingale processes.

In the absence of a probability measure, modelling objects (trajectories, portfolios, stopping times, etc.) are treated here through a robust perspective (\cite{benTal})
and so are subject to relevant optimizations.
The logical constraints imposed by arbitrage related notions can be encoded by the supremum and infimum operators; which, most of the times, are being used to ascertain the existence of an object with the prescribed properties. These operators
also appear when we price options; in this respect,
a minmax perspective can be perceived as too extreme (\cite{sniedovich}) as it considers a worst scenario approach, this view can be deceptive as the meaning of worst scenario is tide up to the functional being optimized and the actual model.
In the case of option pricing,
the functional proposed in BJ\&N is the pathwise error and thus reflects the underlying purpose behind risk neutral pricing, namely pathwise hedging approximation. To provide support for this point of view we show that,
for a discretely attainable option
in a given continuous-time  martingale
market model, the risk neutral pricing can be seen as an example of the minmax pricing described in our paper for an associated discrete model $\mathcal{M}$. We also prove that such a market  $\mathcal{M}$ is $0$-neutral. Here, our approach becomes conceptually close
to model uncertainty, we mention \cite{nutz} (which deals with super-replication and model uncertainty) as a representative of this burgeoning literature. As mentioned, a main difference of our approach is that we give  central stage to the set $\mathcal{S}$, natural hypothesis on this set imply fundamental properties of the pricing functional.
Some minmax publications, with rather different points of view from our paper, with
applications to finance are: \cite{abernethy}, \cite{demarzo} and \cite{rustem}, among other references.

The emphasis of the  paper is to establish basic theoretical properties that follow from the proposed framework. A detailed computational analysis of related examples is available in \cite{degano}; we expect to make clear that
the present setting is quite flexible, several numerical examples and processing of market data can be found in \cite{rahnemaye}.

A summary of the paper contents is provided next.  Section \ref{generalDiscreteModel} introduces trajectory and portfolio sets leading to the trajectory based discrete market models to be used in the paper and remarks on the scope and generality of the framework. Section \ref{allMainConcepts} collects the main definitions and gives some hints of the relevance of these concepts for the rest of the paper. An augmented formalism, allowing for other sources of uncertainty, is also described.
Section \ref{abstractExample} describes an example illustrating the framework.
Section \ref{sec:0-neutral} elaborates on the minmax price bounds, defines option payoffs and a class of minmax functions, playing the role of integrable functions, are introduced as well. 
Section \ref{sec:0-neutralII} proves existence of a price interval $[\underline{V}(Z), \overline{V}(Z)]$ under general $0$-neutrality conditions. That section also compares the price interval with Merton's bounds; Section \ref{sec:pricingIn0-neutralMarkets} describes the meaning of the pricing interval when the market allows for arbitrage. Section \ref{sec:localCharacterizations}  provides general and natural sufficient conditions leading to $0$-neutral and no arbitrage markets. It also introduces concrete market assumptions leading to a price interval.
Section \ref{attainability} deals with attainable functions,
a generalization of this notion and some implications. Some analogues of martingale-like results are proven: in a $0$-neutral market, today's stock
price is the minmax price of future stock prices and we also establish a trajectory based version of the optional sampling theorem. Section \ref{sec:arbitrageFreeMarkets}  provides a general example of a discrete market free of arbitrage such that its trajectory set can not be the support of any martingale.
Section \ref{relationToMartingales} studies a general trajectory based market associated to a continuous-time martingale market model and draws connections between the introduced bounds and risk neutral pricing.
Section \ref{sec:conclusions} concludes. The appendices, collect further results, proofs, as well as some technical results needed in the main body of the paper.

\section{General, Discrete, Trajectory Based  Model} \label{generalDiscreteModel}

The paper concentrates entirely on discrete, non probabilistic,  market models
extending the model  in \cite{BJN}. The setting could be considered
as a discrete version of the non probabilistic, trajectory based, continuous-time models recently introduced in  \cite{AFO2011} and further developed in  \cite{alvarez}.
An example is given in Section \ref{motivationalExample} illustrating a general approach to constructing trajectory sets without using a priori probabilistic assumptions.

\subsection{General Setting}

We now proceed with  formal definitions.

\begin{definition}[Trajectory Set]  \label{trajectories}
Given a real number $s_0$ a set of (discrete) trajectories $\mathcal{S} = \mathcal{S}(s_0, \Sigma)$ is a subset of the
following set
\begin{equation}  \nonumber
\mathcal{S}_{\infty}= \mathcal{S}_{\infty}(s_0)= \{S = \{S_i\}_{i \geq 0}: ~S_i \in \Sigma_i,~~ S_0= s_0\},
\end{equation}
$\Sigma =\{\Sigma_i\}$ is a family of fixed subsets of  $\mathbb{R}$.
\end{definition}
\begin{definition}[Portfolio Set]  \label{locallyDefinedPortfolios}
A portfolio $H$  is a sequence of (pairs of) functions
$H = \{\Phi_i = (B_i, H_i)\}_{i\geq 0}$ with
$B_i, H_i: \mathcal{S} \rightarrow \mathbb{R}$, where
$\mathcal{S} \subseteq \mathcal{S}_{\infty}(s_0)$.
$H$ is said to be self-financing at $S \in \mathcal{S}$ if for all $i \geq 0$
\begin{equation} \label{selfFinancing}
H_i(S) ~S_{i+1} + B_i(S) = H_{i+1}(S) S_{i+1} + B_{i+1}(S).
\end{equation}
A portfolio $H$ is  called non-anticipative if for all $S, S' \in \mathcal{S}$,
satisfying $S'_k = S_k$ for all $0 \leq k \leq i$, it then follows that
 $\Phi_i(S) = \Phi_i(S')$.
\end{definition}

\begin{definition}[Trajectory Based Discrete Market]  \label{discreteMarkets}
For a given  real number $s_0$, a set of trajectories $\mathcal{S} \subseteq \mathcal{S}_{\infty}(s_0)$ and a set of portfolios $\mathcal{H}$, a trajectory based discrete market  $\mathcal{M}$ is a set satisfying the following properties:

\begin{enumerate}
\item $\mathcal{M} = \mathcal{S} \times \mathcal{H}$.
\item For each $(S, H) \in \mathcal{M}$ there exist an integer $N = N_H(S)$, such that

\hspace{1.5in} $[H_{k}(S) = H_{N-1}(S),~~~\forall k \geq N_H(S)]$ or $[ H_{k}(S) = H_{N}(S) =0, ~~~\forall k \geq N_H(S) ]$.
\item For $(S, H) \in \mathcal{M}$, $H$ is non-anticipative and  self-financing at $S$.
\end{enumerate}
\end{definition}
Let $ H = 0 = \{(0_i,0_i)\}_{i\geq 0}$ (where $0_i$
are the function $0_i(S) =0$) denote the $0$-portfolio;
for any discrete market $\mathcal{M}$ we will assume $\{H=0\} \in \mathcal{H}$, with $N_0 \equiv 1$.

$H_{k}(S) = H_{N-1}(S)$ for all $k \geq N_H(S)$ means rebalancing stops
at, or prior to, $N_H(S)-1$. The condition $H_{k}(S) = H_{N}(S)=0$ for all $k \geq N_H(S)$ means definite liquidation has taken place at, or prior to, $N_H(S)$;
such portfolio will be referred to as {\it liquidated}.

Shortly, we will extend the above setting to account for {\it other sources of uncertainty}, accommodating this extension is mostly a matter of notation and, hence, most of the paper will only employ the above introduced notation.

The mathematical definition of market model $\mathcal{M}$, when applied to an unfolding stock chart $S(t)$ and bank account $B(t)$, uses the following obvious interpretations.
The numbers $H_i(S)$ and $B_i(S)$ are interpreted, respectively, as the holdings on the stock and the balancing bank account {\it just after} the $i$-th. trading has taken place.
$S_i$ is the value taken by the unfolding chart at the $i$-th trading. To summarize: the portfolio values $(B_i(S), H_i(S))$ are held in the trading period
$(i, i+1]$, the definition of $\mathcal{M}$ includes trajectories and portfolio re-balancing, a trajectory dependent
number of times, until the position in the stock is liquidated or rebalancing stops. When valuing options $N_H(S)$ will be an instance before the
European option expires.

Given $(S, H) \in \mathcal{M}$, the self-financing property (\ref{selfFinancing})
implies that the portfolio value, defined  by $V_{H}(i, S) = B_i(S) + H_i(S) ~ S_i$ equals:
\begin{equation}  \label{portfolioValue}
V_{H}(i, S) = V_{H}(0, S_0) + \sum_{k=0}^{i-1} H_k(S) ~(S_{k+1} - S_{k}),
\end{equation}
during the period $[i, i+1)$ for $i =0, \ldots, N_H(S)-1$ and valid over $[N_H(S) , \infty)$ for the case
$i \geq N_H(S)$. Of course, $V_{H}(0, S_0) \equiv V_{H}(0, S) = B_0(S) + H_0(S) ~ S_0$.

Observe that, for simplicity we have assumed in last equation and in (\ref{selfFinancing}), and it will remain in the sequel, that the interest rate of the bank account is zero, and that there are no transaction costs.
\begin{remark}  \label{alternativeConstrOfSFPort}
As defined above, a portfolio $H$ is given by specifying the pairs of  functions $\{(B_i, H_i)\}$ so that (\ref{selfFinancing}) holds. In the remaining of the paper, we will define $H$ more conveniently by specifying the non-anticipative functions $H_i$ and an initial portfolio value $V_0= V_{H}(0, S_0)$, this will provide $B_0$, the remaining functions $B_i$, $i \geq 1$,  are then obtained by solving equations (\ref{selfFinancing}).
\end{remark}

The above definitions are natural generalizations of the ones introduced in \cite{BJN}
(see also the book presentation of the material in \cite{rebonato}). The definitions make explicit the notion
of market model by formalizing the notion of set of portfolios (left out informal in \cite{BJN}).

Informally, we explain the rather general nature of the above introduced framework.
Notice that nothing requires $H_i(S) \neq H_{i+1}(S)$, in particular, actual rebalancing of the stock holdings
could have stopped well before $N_H(S)-1$. $S_{i+1}$ is the stock
value at which investors $H \in \mathcal{H}$, that have invested so far $H_k(S)$, $0 \leq k \leq i$,
may rebalance their holdings to $H_{i+1}(S)$. The set of values $\Sigma_i$  taken by the stock components $S_i$ can be an arbitrary fixed subset of $\mathbb{R}$, for example, values of $S_i$ could be represented by a finite number of decimal digits. Similarly, the values $H_i(S)$ can
belong to an arbitrary fixed subsets of $\mathbb{R}$, for example, integer multiples of a given real number.


\vspace{.1in}\noindent Some results require that the functions $N_H:\Se \rightarrow \mathbb{N}$, introduced in Definition \ref{discreteMarkets}, are stopping times, according to the following definition.
\begin{definition}[Trajectory Based Stopping Times]\label{stoppingTimeDef}
Given a trajectory space $\mathcal{S}$ a {\it trajectory based stopping time} (or {\it stopping time} for short) is a function $\nu: \mathcal{S} \rightarrow \mathbb{N}$ such that:
\[\mbox{if}\quad S, S' \in \mathcal{S}\quad \mbox{and}\quad S_k=S'_k,~~~ \mbox{for} ~~~ 0 \leq k \leq \nu(S),\quad \mbox{then}\quad \nu(S') = \nu(S).\]
\end{definition}
\noindent

We refer to \cite{shiryaev} (see also \cite{alvarez}) for an account of the relationship between the above notion of trajectory based stopping time and filtration based stopping times.

\section{Global, Conditional and Local Concepts}   \label{allMainConcepts}

This section collects most of the basic concepts needed in the remaining of the paper and makes comments on their relevance and interrelationship.

\vspace{.1in}
Definition \ref{priceBounds} below provides  fundamental, worst case, pricing definitions by means of a global minmax optimization; they were introduced in \cite{BJN} in the context of trajectory based markets. Section \ref{sec:0-neutral} provides results showing the import of the minmax bounds, other properties are relegated to Appendix \ref{furtherPricingResults}.

\begin{definition}[Price Bounds]  \label{priceBounds}
Given a discrete market $\mathcal{M} = \mathcal{S} \times \mathcal{H}$ and a function $Z: \mathcal{S} \rightarrow \mathbb{R}$,
define the following quantities:
\begin{equation}  \label{upperBound}
\Vup(S_0, Z, \mathcal{M})= \inf_{H \in \mathcal{H}} \left\lbrace \sup_{S \in \mathcal{S}} \left\lbrace Z(S)-\sum_{i=0}^{N_H(S)-1} H_i(S)~\Delta_iS \right\rbrace \right\rbrace,
\end{equation}
and
\begin{equation}  \nonumber \label{lowerBound}
\Overline(S_0, Z, \mathcal{M})= \sup_{H \in \mathcal{H}} \left\lbrace \inf_{S \in \mathcal{S}} \left\lbrace Z(S) +\sum_{i=0}^{N_H(S)-1} H_i(S)~\Delta_iS \right\rbrace \right\rbrace.
\end{equation}
Clearly, $\Overline(S_0, Z, \mathcal{M}) = - \Vup(S_0, -Z, \mathcal{M})$.
\end{definition}

Notice that $\overline{V}$ and $\underline{V}$ are monotonic functions of $Z$.
Essentially, the quantity $\Vup(S_0, Z, \mathcal{M})$ is the smallest initial capital $V_0$ such that there exists a portfolio in $\mathcal{H}$ that, when used along with this initial capital, will upper-hedge the function $Z$ uniformly on the trajectory space $\mathcal{S}$. Similarly, $\Overline(S_0, Z, \mathcal{M})$ is the largest initial capital  such that there exists a portfolio in $\mathcal{H}$ that, when used along with this initial capital, will lower-hedge the function $Z$ uniformly. The precise statements are provided in Propositions \ref{superhedgeProposition} and \ref{tightSuperhedgeAndUnderhedge} in Section \ref{minimaxPricing}.

\vspace{.1in} The next definition is the notion of arbitrage used in the paper.

\begin{definition}[Arbitrage-Free Market] \label{ArbitrageDefinition}
Given a discrete market $\mathcal{M}$, we will call $H \in \mathcal{H}$ an arbitrage strategy if:
\begin{itemize}
\item $\forall S \in \mathcal{S}$,  $V_{H}(N_H(S), S)\geq V_{H}(0, S_0)$.
\item $\exists S^{\ast} \in \mathcal{S}$ satisfying $V_{H}(N_H(S^{\ast})    , S^{\ast})) > V_{H}(0, S_0)$.
\end{itemize}

We will say  $\mathcal{M}$ is arbitrage-free if $\mathcal{H}$ contains no arbitrage strategies.
\end{definition}
It is customary to add the extra condition $V_{H}(0, S_0) \leq 0$, by {\it not} imposing the constraint  $V_{H}(0, S_0) \leq 0$, Definition \ref{ArbitrageDefinition} reflects the fact that one could make a profit without risk even though an initial positive capital may be involved.

For $S\in \mathcal{S}$  we will use the notation $\Delta_iS \equiv S_{i+1}-S_i$ for $i\ge 0$.
Whenever convenient, the tuple $(S,k)$ or the triple $(S,H,k)$ will be referred generically as {\it a node}.

\subsection{$0$-Neutral Markets}
Consider the function $Z\equiv 0$. Since the null portfolio $H\equiv 0$ belongs to $\mathcal{H}$, it results that $\Vup(S_0, 0, \mathcal{M})\le 0$. Then, Proposition \ref{superhedgeProposition} from Section \ref{sec:0-neutral}, indicates that no  positive or negative number $\pi$ could be a fair price (this notion is introduced in Defininition \ref{relativeArbitrage}, Section \ref{minimaxPricing})  for $Z=0$ as those situations will create a relative riskless profit. So, $\pi=0$ should be the unique fair price for the function $Z\equiv 0$, this imposes a restriction on the market which leads to the next definition.
\begin{definition}[$0$-Neutral Market] \label{0NeutralDefinition}
A discrete market $\mathcal{M}$ is called $0$-neutral if
\begin{equation} \nonumber 
\inf_{H \in \mathcal{H}} \left\lbrace \sup_{S \in \mathcal{S}} \left\lbrace -\sum_{i=0}^{N_H(S)-1}~H_i(S)~\Delta_iS \right\rbrace \right\rbrace=0.
\end{equation}
\end{definition}
\noindent
Notice that $0$-neutrality means $\overline{V}(S_0, Z=0, \mathcal{M})= 0 =\underline{V}(S_0, Z=0, \mathcal{M})$.
It is easy to see that an arbitrage-free market is $0$-neutral (Corollary \ref{necessaryCondition},  Section \ref{sec:localCharacterizations}); it should also be clear that a general $0$-neutral market allows for arbitrage, a brief discussion is presented following Corollary \ref{necessaryCondition} in Section \ref{sec:localCharacterizations}.

\vspace{.1in}
\noindent The following conditional spaces will play a key role.
Given $\mathcal{M}$ and for $S  \in \mathcal{S}$ and $k \geq 0$ fixed, set:
\begin{equation} \nonumber
{\mathcal S}_{(S,k)}\equiv\{\tilde{S} \in \mathcal{S}: \tilde{S}_i= S_i, 0 \le i \le k\}.
\end{equation}
The multiplicity of these sets indicate the incomplete nature of the markets that we are
introducing. The analogue to the sets ${\mathcal S}_{(S,k)}$ in stochastic models are, in general, sets of measure zero.

We will need to generalize the above notions of minmax bounds to contemplate the possibility of conditioning
on given values of $S$ and trading instance $k$. We present the basic definitions next.

\begin{definition}[Conditional Minmax Bounds]
Given a discrete market $\mathcal{M} = \mathcal{S} \times
\mathcal{H}$, $S \in \mathcal{S}$ as well as an integer $k \ge 0$,
define
\begin{equation} \nonumber
 \overline{V}_k(S, Z, \mathcal{M}) \equiv ~\inf_{H \in \mathcal{H}}~\sup_{\tilde{S} \in \mathcal{S}_{(S, k)}} [Z(\tilde{S}) - \sum_{i=k}^{N_H(S)-1} H_{i}(\tilde{S})
 \Delta_i\tilde{S}].
\end{equation}
Also define $\underline{V}_k(S, Z, \mathcal{M}) = -
\overline{V}_k(S, -Z, \mathcal{M})$.
\end{definition}

\noindent
Notice that $\overline{V}_0(S, Z, \mathcal{M})= \overline{V}(S_0,
Z, \mathcal{M})$ and so, $\underline{V}_0(S, Z, \mathcal{M})=
\underline{V}(S_0, Z, \mathcal{M})$ as well.

\begin{definition}[Conditionally $0$-Neutral]  \label{conditionally0Neutral}
We say that a discrete market $\mathcal{M}$ is \emph{conditionally
$0$-neutral} at $S \in \mathcal{S}$, and $k \geq 0$, if
\begin{equation}  \nonumber 
 \overline{V}_k(S, Z=0, \mathcal{M}) =0.
\end{equation}
\end{definition}
\noindent
Observe that, for $k=0$, the conditional $0$-neutral property, which
depends on $S$ only through $S_0$, reduces to $0$-neutral.

\subsection{Local Notions}  \label{localNotions}

The next definition introduces two basic concepts: a local, and  portfolio independent, analogue on $\mathcal{S}$ of the $0$-neutral property of $\Me$ and a strengthening of this notion representing the local analogue of the arbitrage-free property. These local notions are instrumental as conditions on $\mathcal{S}$  ensuring  $\mathcal{M}=\mathcal{S}\times\mathcal{H}$ to be $0$-neutral or arbitrage-free (but conditions on $\mathcal{H}$, through $N_H$, are needed as well, See Section \ref{sec:localCharacterizations}).
\begin{definition} [$0$-Neutral \& Arbitrage-Free Nodes] \label{localDefinitions}
Given a trajectory space $\Se$ and a node $(S,j)$:

\begin{itemize}
\item $(S,j)$ is called a \emph{$0$-neutral node} if
\begin{equation} \label{cone}
\sup_{\tilde{S} \in \Se(S, j)}~~ (\tilde{S}_{j+1} - S_{j}) \geq    0~~~~~\mbox{and} ~~~~~
\inf_{\tilde{S} \in \Se(S, j)} ~~(\tilde{S}_{j+1} - S_{j}) \leq    0.
\end{equation}

\item $(S,j)$ is called an \emph{arbitrage-free node} if
\begin{equation} \label{upDownProperty}
\sup_{\tilde{S} \in \Se(S, j)}~~ (\tilde{S}_{j+1} - S_{j}) >   0 ~~~~~\mbox{and} ~~~~~
\inf_{\tilde{S} \in \Se(S, j)} ~~(\tilde{S}_{j+1} - S_{j}) <   0 \\
\end{equation}
or
\begin{equation} \label{flatNode}
 \sup_{\tilde{S} \in \Se(S, j)}~~ (\tilde{S}_{j+1} - S_{j}) =
\inf_{\tilde{S} \in \Se(S, j)} = 0 = (\tilde{S}_{j+1} - S_{j}).
\end{equation}
\end{itemize}
\noindent $\Se$ is called \emph{locally $0$-neutral} if (\ref{cone}) holds at each node $(S, j)$. $\Se$ is said to be \emph{locally arbitrage-free}  if either (\ref{upDownProperty}) or (\ref{flatNode}) hold at each node  $(S, j)$.

A node that satisfies (\ref{upDownProperty}) will be called an up-down node, and a node satisfying  (\ref{flatNode}) will be called a {\it flat node}. A node that is $0$-neutral but that is not an arbitrage-free node, will be called an {\it arbitrage node}.
\end{definition}
\begin{remark}
The developments need to distinguish related notions applicable to $\mathcal{S}$, $\mathcal{M}$ or to actual nodes (e.g. $\mathcal{M}$ is arbitrage-free, $\mathcal{S}$ is locally arbitrage-free, etc.). To help avoiding confusion we may use the words {\it global} when referring to properties of $\mathcal{M}$
and  {\it local} when referring to properties of $\mathcal{S}$.
\end{remark}
\noindent  An arbitrage-free node is clearly a  $0$-neutral node as well. If all nodes
$(\tilde{S}, k)$, $\tilde{S} \in \mathcal{S}_{(S, j)}$ and $k \geq j$,  are $0$-neutral  and $\He$ is a set of portfolios, it follows that:
\begin{equation} \label{0NeutralPropertyUsefulForInduction}
\inf_{H \in \He} \{ \sup_{\tilde{S} \in \Se_{(S,j)}}[-H_j(S)\Delta_j\tilde{S}\,]\}=0.
\end{equation}

The local arbitrage-free property of $\mathcal{S}$  plays the analogous role to the  martingale property in a stochastic setting, with this in mind one can try to prove martingale-type results. We provide one such example with a trajectory based version of the optional sampling theorem in Section \ref{attainability}.





\vspace{.1in} Given $\mathcal{M} = \mathcal{S} \times \mathcal{H}$, $S^*\in \Se$ and assuming  $N_H$ to be  bounded  for each $H \in \mathcal{H}$, the following results hold:
\begin{enumerate}
\item If all nodes $(S,j)$ in $\Se_{(S^*,k)}$, $j \geq k$, are  $0$-neutral then,
$\mathcal{M}$ is conditionally $0$-neutral at $(S^*, k)$.

\item If all nodes $(S,j)$ in $\Se$, $j \geq 0$ are arbitrage-free and $N_H$ is a stopping time, then
$\mathcal{M}$ is arbitrage-free.
\end{enumerate}
Item $(1)$ follows as a special case of Theorem \ref{0-neutral} and  item $(2)$ is a special case of Corollary \ref{upDownCorollary} (both results are found in Section \ref{sec:initiallyBounded}). We note that these results hold more generally for cases when $N_H$ is not necessarily bounded.

\subsection{Other Sources of Uncertainty}  \label{sec:otherSourcesOfUncertainty}

All results and definitions in the paper involving markets $\mathcal{M}$ and trajectory sets $\mathcal{S}$
can be generalized by incorporating another source of uncertainty besides
the stock. This extra source of uncertainty  will be denoted by $W = \{W_i\}$ which,  in financial terms, will be considered  to be an observable quantity. This is analogous to moving from the natural filtration to an augmented filtration in the stochastic setting.

The sequence elements $W_i$ are assumed to belong to abstract sets $\Omega_i$ from which we only require to have defined an equality relationship. We provide next the simple changes to the previous definitions
to accommodate for the new source of uncertainty. The arrow notation $\rightarrow$ indicates
how the objects change ($(s_0, w_0)$ is fixed).
\begin{align} \label{generalSet}
\mathcal{S}_{\infty}(s_0) &\rightarrow \mathcal{ S}^{\mathcal W}_{\infty}(s_0, w_0) \equiv \{{\bf S} = \{{\bf S}_i \equiv (S_i, W_i)\}_{i \geq 0}: S_i \in \Sigma_i\subset \mathbb{R}, W_i \in \Omega_i~ S_0 = s_0, W_0 = w_0\}.  \\ \nonumber
\mathcal{S} = \mathcal{S}(s_0) &\rightarrow \mathcal{S}^{\mathcal W}(s_0, w_0)
\subseteq \mathcal{S}^{\mathcal W}_{\infty}(s_0, w_0).\\  \nonumber
H_i(S) &\rightarrow H_i({\bf S})\\
\mathcal{S}_{(S, k)}  &\rightarrow \nonumber \mathcal{ S}^{\mathcal
W}_{({\bf S}, k)}(s_0, w_0) \equiv \{{\bf \tilde{S}} \in
\mathcal{S}^{\mathcal W}(s_0, w_0),~{\bf \tilde{S}}_i = {\bf
S}_i,~~~~ 0 \leq i \leq k\}.\\ \nonumber V_H(i, S) &\rightarrow
V_H(k, {\bf S}) = V_H(0, (S_0, w_0)) + \sum_{i=0}^{k-1} H_i({\bf S}) \Delta_i S.\\
\nonumber
\end{align}
Besides the above changes, that concern mostly trajectory sets and the
functional dependency $H_i({\bf })$ in terms of both variables $S_k, W_k$
(and some minor  notational changes),
all statements and  properties appearing in the paper, only involve the first coordinate $S_i$ (in the tuples $(S_i, W_i)$) in all algebraic manipulations. Clearly, $H_i$ is required to be non-anticipative with respect
two both variables $S_k$ and $W_k$ and the notion of trajectory based stopping time applies now to trajectories of the form ${\bf S} = \{(S_i, W_i)\}$.  These remarks can be used to show that all the results in the paper stay true in the extended/augmented formalism.
We explicitly use the extended formalism in subsection \ref{abstractExample} and Section \ref{relationToMartingales}.


\section{Example}  \label{motivationalExample}
To motivate and illustrate discrete markets $\mathcal{M}=\mathcal{S}\times\mathcal{H}$, we introduce a family of examples. These examples model discretizations of stock charts which, {\it for the moment}, are assumed to be given as a family of continuous-time functions $\mathcal{X}(x_0) \subseteq \mathcal{X}_{\infty}(x_0) \equiv \{ x\in \mathbb{R}^{[0,T]}: x(0)=x_0\}$, with $x_0$, $T >0$ and $s_0= e^{x_0}$. We will rely on some defininitions.

\vspace{.05in}\noindent {\it Refining Sequence of Partitions:}  Consider a sequence $\{\Pi_n\}_{n\ge 1}$, where $\Pi_n= \{r^n_i\}$ is a finite partition of $[0, T]$ with $r^n_0 =0$ and $\Pi_{n} \subseteq \Pi_{n+1}$. Let $\Pi \equiv \cup_{n\ge 1} \Pi_n$.

\vspace{.05in}\noindent {\it Selected times:} Let ${\bf t}=\{t_i\}_{i\ge 0}$ a sequence of functions $t_i:\mathcal{X}(x_0) \rightarrow \Pi$ such that $t_0=0$,
\[\forall x\in \mathcal{X}(x_0)\;\; \exists m(x)\in \mathbb{N}:\;\; t_{i}(x) < t_{i+1}(x),\;\;\mbox{if}\;\;0\le i < m(x), \;\;\; t_i(x)=T\;\;{if}\;\;i\ge m(x),\;\;\mbox{and}\]
\[\forall i\ge 1,\;\;\mbox{if}\;\;\tilde{x},x\in \mathcal{X}(x_0)\;\;\mbox{with}\;\; \tilde{x}(s)=x(s)\;\;\mbox{for}\;\; 0\le s \le t_i(x)\;\;\mbox{then}\;\; t_i(\tilde{x})=t_i(x).\]
Observe that for any $x\in \mathcal{X}(x_0)$ there exist $n\ge 1$ such that $\{t_i(x)\}_{i\ge 0} \subset \Pi_n$. This is so because $t_i(x)\in \Pi$ implies that there exists $n_i$, the minimum such that $t_i(x)\in \Pi_{n_i}$, and $n=\max \{n_i : 0\le i\le m(x)\}$.

Let $\Ref=(\Pi,{\bf t})$, define the following  general class of discrete trajectories
\begin{equation} \nonumber 
\Se(s_0, \mathcal{X},\Ref)\equiv \{S=\{S_i\}_{i\ge 0}:S_i=\exp(x(t_i)),\;\;i\ge 0,\,\,\mbox{for some}\;x\in \mathcal{X}(x_0)\}\quad \small {(t_i=t_i(x))}.
\end{equation}

\vspace{.1in}\noindent {\bf General Aspects.} A refining sequence of partitions reflects a financial situation where the investor re-balances her/his portfolio with a certain minimum time resolution but is willing to refine it further if deemed necessary. The case of a
fixed partition ($\Pi_n$ the same for all $n$) means that the investor will never rebalance more often than an a-priori  given time resolution.

There is no essential result in our paper that requires $S_i\geq 0$, so there is no need to use the exponential function in the definition $S_i = e^{x(t_i)}$ but, doing so makes it easier to connect with the usual geometric stochastic models as well as with \cite{BJN}.

We are interested in prescribing ``structured" subsets of $\mathcal{S}(s_0, \mathcal{X}, \Ref)$, we do this
by means of an observable functional $F$. For simplicity, the functional could be defined on
$\mathcal{X}(x_0)$ (and could depend on other variables as well) and takes values on $\mathbb{R} \cup \{\infty\}$. As a particular case, the functional $F$ selects those $x\in \mathcal{X}(x_0)$ such that $x$ has finite quadratic variation in the interval $[s,t]$ respect to $\Pi$, that is,  the following limit exists
\begin{equation} \label{qvFunctional}
F(x, s, t) = \lim_{n \rightarrow \infty} \sum_{s \leq r^n_{i}, ~~
r^n_{i+1} \leq t}~~(x(r^n_{i+1}) - x(r^n_{i}))^2.
\end{equation}
Other observable, additive, non-decreasing and non-negative quantities could be used as well, for example, the number of transactions from $t_0=0$ to $t$, or the total number of shares transacted from $t_0=0$ to $t$. Intuitively, each time a transaction takes place,
the value for the observable quantity represents the tick of a trajectory based clock (usually interpreted as  a ``business clock" with trajectory dependent rate.)

The following is an example of a structured discrete trajectory set. For $c,~ d>0$ real numbers and $Q\subset(0,\infty)$, define:
\begin{align} \label{classOfCharts}
\mathcal{S}(s_0, c, d, Q)&=\{S \in \mathcal{S}(s_0, \mathcal{X}, \Ref): |x(t_{i+1}) - x(t_i)| \leq d,
~~F(x, t_i, t_{i+1}) \leq c,~~F(x, t_0, t_{m(x)}) \in Q,~~0\leq i < m(x)\}.
\end{align}
For the case of the functional given by (\ref{qvFunctional}), the requirements defining $\mathcal{S}(s_0, c, d, Q)$ can be interpreted as imposing constraints on the maximum consumed quadratic variation and on the maximum absolute value
of the change on chart values, both in between consecutive trading instances.  In addition, the condition $F(x, t_0, t_{m(x)}) \in Q,$ means we deal with trajectories whose total quadratic
variation in the interval $[0, T]$ belongs to the a-priori given
subset $Q$. In effect, the imposed constraints restrict the outcomes resulting from the interaction between market fluctuations and portfolio rebalances.


\subsection{Construction of Trajectory Sets From Augmented Data}
\label{abstractExample}

Here we describe a set of trajectories that does not require
a continuous time model for the charts. The general principle guiding the construction is to isolate an observable quantity (representing a variable of interest) and proceed to define a trajectory space by imposing constraints relating the trajectories and a
free variable representing this observable. We work with observables given by a
functional, still denoted by $F$, but now defined on {\it charts} (i.e. the values of some unfolding financial data), $F$  may also depend on other variables as well.

The definition of $\mathcal{S}(s_0, c, d, Q)$ in (\ref{classOfCharts}) depends on having access to the functions $x\in \mathcal{X}$. We now  turn the tables around and re-define $\mathcal{S}(s_0, c, d, Q)$ as $\Swe(s_0, c, d, Q)$, a set which does not require any reference to a given class of continuous-time trajectories.
We still allow observable charts to unfold in continuous-time, in order to achieve these goals, we use the augmented formalism introduced in Section \ref{sec:otherSourcesOfUncertainty}. Trajectories are given by a sequence of tuples ${\bf S} = \{(S_i, W_i)\}_{i\ge 0}$, which will be associated to samples of continuous-time charts. This is a natural way to proceed, information deemed
important for modelling is lost when sampling hence, this information will be  encoded by the variable $W_i$ (associated to the functional $F$).


The definition below assumes given:  $w_0=0$,  $s_0$ and $c,~ d>0$ real numbers, $\sum_i \subseteq \mathbb{R}$ and sets $Q, \Omega_i \subset(0,\infty)$.

\begin{definition} \label{implicitDefinitionThroughConstraints}
$\mathcal{S}^{\mathcal W}(s_0, d, c, Q)$ will denote a subset of $\mathcal{S}^{\mathcal W}_{\infty}(s_0, w_0)$ (this last set as in (\ref{generalSet})) so that ${\bf S} \in \mathcal{S}^{\mathcal W}(s_0, d, c, Q)$  satisfies $S_i\in \Sigma_i$, $W_i \in \Omega_i$ and:
\begin{enumerate}  \label{generalConstraints}
\item $|\log S_{i+1} - \log S_i| \leq d$ for all $i \geq 0$,
\item $0 < W_{i+1} - W_i \leq c$ for all $i \geq 0$.\\
Moreover, there exists at least one $i^{\ast}$ such that
\item   $W_{i^{\ast}}  \in Q$.
\end{enumerate}
Associated discrete markets $\mathcal{S}^{\mathcal W}(s_0, c, d, Q) \times \mathcal{H}$ are required to satisfy: $H \in \mathcal{H}$ then
$W_{N_H({\bf S})} \in Q$.
\end{definition}
\begin{remark} \label{directlyModelingTheStock}
As already mentioned, the condition $|\log S_{i+1} - \log S_i| \leq d$ could equally be replaced by $|S_{i+1} -  S_i| \leq d$ (of course with an appropriately chosen value for $d$).
\end{remark}

We emphasize that $\mathcal{S}^{\mathcal W}(s_0, d, c, Q)$, as characterized above, does not need to be, in general, the set of {\it all} trajectories ${\bf S}$ satisfying the listed constraints in Definition (\ref{generalConstraints}); specific examples  are described in \cite{degano}. Comparing with (\ref{classOfCharts}), we see that we have allowed $F$ to be an independent variable $W$. For ${\bf S} \in \mathcal{S}^{\mathcal W}(s_0, d, c, Q)$ there could be multiple indexes $i^{\ast}$.

The set $\mathcal{S}^{\mathcal W}(s_0, c, d, Q)$ is used for modelling
the unfolding of a data chart $x(t_i)$ by mapping $\{(e^{x(t_i)},F(x,t_0,t_i))\}$, one index $i$ at a time (i.e. as the chart unfolds), to its closest path $\{(S_i, W_i)\}_{i\geq 0}$.

The trajectory set introduced in \cite{BJN} can be recovered as a
special case of Definition \ref{implicitDefinitionThroughConstraints} by
taking $Q = \{v_0\}$ and defining
\begin{equation}  \label{sampledQV}
W_i = \sum_{k=0}^{i-1} (\log S_{k+1} - \log S_k)^2,
\end{equation}
moreover we need to require the existence of $i^{\ast}$ satisfying $W_{i^{\ast}}= v_0$.
Therefore $W_{i+1}  - W_i = (\log S_{i+1} - \log S_i)^2$ and the
constraint $0 < W_{i+1} - W_i \leq c$  in Definition
\ref{implicitDefinitionThroughConstraints} corresponds to $c=d^2$. Moreover,  as  $W_i$ depends on
$S_k,~0 \leq k \leq i,$ there is no need to work with tuples $(S_i, W_i)$ in this
case.  Not imposing
(\ref{sampledQV}) allows to
incorporate $0$-neutral nodes which are arbitrage nodes (see Definition \ref{localDefinitions} and related comments afterwards.) An analysis of these considerations in the context of the example is outside the scope of the paper, details are given in \cite{degano}.

A natural discretization leading to an  implementation of $\Swe\!\!(s_0, c, d, Q)$ is obtained by introducing real numbers $\delta, \beta >0$. The
coordinates $S_i$ are then restricted to belong to the sets
$\Sigma_i=\Sigma(\delta) \equiv \{s_0 e^{k \delta},~~ k \in
\mathbb{Z}\}$ and $W_i$ to $\Omega_i \equiv \Omega(\beta) =
\{ j \beta^2, ~j \in \mathbb{N}\}$, thus $Q$ is now a set $Q(\beta) \subseteq \Omega(\beta)$.

For implementation purposes we need a finite version of the above discrete space. Towards this end we will take, for convenience,
$\frac{d}{\delta} =p$ a fixed integer. So, the jump bound $d$ is given as an integer multiple of $\delta$. For given $N_1, N_2 \in
\mathbb{N}$ define
\begin{equation} \nonumber
 \Sigma(\delta, N_1) = \{s_0e^{k~\delta}, k \in \{-N_1, -N_1 +1, \ldots, N_1\}\},\qquad \Omega(\beta, N_2)  =\{j~\beta^2, 0 \leq j\leq N_2\},
\end{equation}
and for a finite $m-tuple$ of positive integers $\Lambda=(n_1,\ldots, n_m)$, $n_j\le N_2$, define
\begin{equation} \nonumber
Q \equiv Q(\beta, \Lambda) = \{ n_k ~\beta^2: 1\leq k \leq m\}.
\end{equation}

We then obtain a finite version of Definition \ref{implicitDefinitionThroughConstraints}, assuming that
\[ S_i\in \Sigma(\delta, N_1)\quad\mbox{and}\quad W_i \in \Omega(\beta, N_2),\quad\mbox{for all}\quad i\ge 0.\]
Such finite versions of the sets $\mathcal{S}^{\mathcal{W}}(s_0, d, c, Q)$ will be denoted $\Se^{\mathcal{W}}(s_0, p, q, \Lambda, N_1, N_2, \delta, \beta)$.

{\bf Local Behavior.} The way of defining trajectory sets $\Swe(s_0,d,c,Q)$, or their finite versions, will make it easy to check if the local properties of $0$-neutral or up-down are satisfied. This is so because our constraints are given locally (i.e. at each node) and the combinatorial definitions will allow trajectories to move up or down. Next, as examples, we provide some general arguments on how to argue for the validity of these local properties.

Assume that the the sets $\Sigma_i$ in the trajectory space $\Swe(s_0,c,d,Q)$ of Definition \ref{implicitDefinitionThroughConstraints} do not attain minimum nor maximum and fix a node $(S_i,W_i)$ of a trajectory ${\bf S}$. Clearly, there exists the possibility of choosing trajectories $\tilde{{\bf S}},{\bf \hat{S}} \in \Swe_{({\bf S},i)}$ such that $\tilde{S}_{i+1}>S_i$, and $\hat{S}_{i+1}<S_i$ respectively, so any node is up-down, and in that case the market results locally arbitrage-free, see Definition \ref{localDefinitions}.
Specific instances of the sets $\Swe(s_0,c,d,Q)$ or their discrete or finite versions will impose further constraints beyond the ones listed in Definition \ref{implicitDefinitionThroughConstraints}. In each such case, 
we will need to check the validity of the needed local requirements, $0$-neutrality or arbitrage-free, so that our results hold.

Consider now $\Se^{\mathcal{W}}(s_0, p, q, \Lambda, N_1, N_2, \delta, \beta)$, in this case we can assume that $N_2 \beta^2 \in Q$. First note, as previously indicated,  that for any trajectory ${\bf S}=\{(S_i,W_i)\}_{i\ge 0}$,
since $W_{i^{\ast}}\in Q \subset \Omega(\beta,N_2)$, then $i^{\ast} \le N_2$.
Taking into account the constraint $d\ge |\log S_{i+1}-\log S_i|=|k_{i+1}-k_i|\delta$,
the largest value that $S_i$ can attain corresponds to the value $S_{N_2}$ and, in that case, $S_{N_2}=s_0e^{N_2\,p\delta}$, which shows  $N_1\le p~N_2$.


In the case that $N_1\le (N_2-1)\,p$, could exist trajectories containing the node
$(S_{i'}=s_0e^{N_1\,\delta}, i'\beta^2)$ with $i'\le N_2-1$ (for example, by choosing  $k_{i}=i\,p$ for $i>0$, when $i\le \frac{N_1}p$ and at most $k_{N_2-1}=N_1$). Such trajectory satisfies $W_{i}\le(N_2-1)\beta^2$ and so, one more step is available. Moreover, given that for any trajectory $\tilde{{\bf S}}\in \Swe_{({\bf S},i')}$ it follows that $\tilde{S}_{i'+1}\le \tilde{S}_{i'}=S_{i'}=s_0e^{{i'}p\delta}$, $({\bf S},{i'})$ is an arbitrage node.  These nodes present arbitrage opportunities.

For display purposes, consider a finite space $\mathcal{S}^{\mathcal{W}}(s_0, c, d, Q)$ consisting of {\it all}  trajectories satisfying the conditions in Definition \ref{implicitDefinitionThroughConstraints} with $N_2=N_1= 100$,  $\beta=\delta=0.0082$, $Q=\{N_2\beta^2\}$, $c=d^2$, $s_0=1$, $p=3$ (and so   $d=0.0246$). Figure \ref{fig:10} shows $200$ random samples of trajectories from such trajectory space. Figure \ref{fig:13} shows random samples of trajectories from two conditional spaces from the above trajectory space.

\begin{figure}[!ht]
\centering
\includegraphics[width=0.6\textwidth]{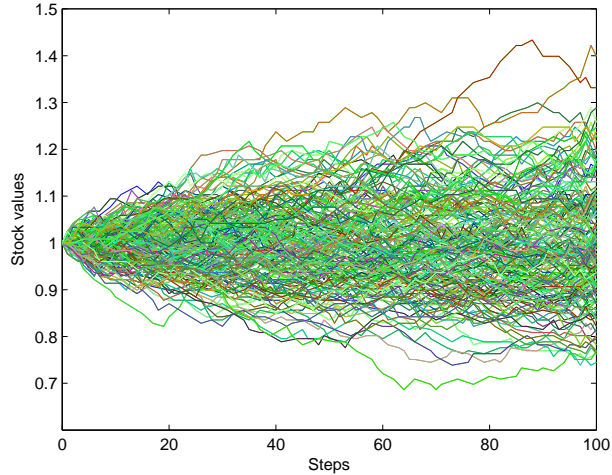}
\caption{Samples from trajectory space.}
\label{fig:10}
\end{figure}

\begin{figure}[!ht]
\centering
\includegraphics[width=0.6\textwidth]{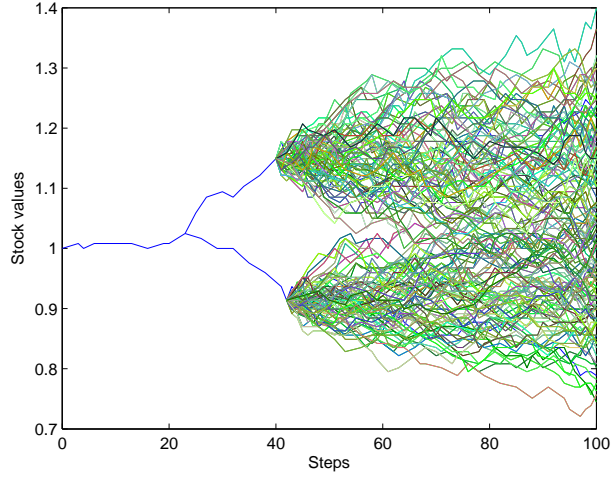}
\caption{Sampled trajectories from conditional spaces.}
\label{fig:13}
\end{figure}

\section{Minmax Bounds}  \label{sec:0-neutral}

Given a  future profile $Z(S)$, 
Definition \ref{priceBounds}  provide the ``price" bounds for the associated option. The present section develops basic results  following from the definitions while Section \ref{sec:0-neutralII} justifies the quantities introduced to be actually price bounds. Options and minmax functions are introduced as well.

\subsection{Minmax Bounds} \label{minimaxPricing}

The price bounds can be recast in a more familiar way:

\begin{equation}  \nonumber 
\Vup(S_0,Z, \mathcal{M}) =  \inf \{u: \exists H \in \mathcal{H},~~~
u +  \sum_{i=0}^{N_H(S)-1} H_i~(S)~\Delta_iS \geq Z(S)~\forall S \in \mathcal{S}\},
\end{equation}
\begin{equation} \nonumber
\underline{V}(S_0,Z, \mathcal{M}) = \sup \{u: \exists H \in \mathcal{H},~~~
u +  \sum_{i=0}^{N_H(S)-1} H_i~(S)~\Delta_iS \leq Z(S)~\forall S\in \mathcal{S}\},
\end{equation}
under the conventions that $\inf \emptyset = \infty$ and $\sup \emptyset = -\infty$.

The notion of relative arbitrage (see \cite{fernholz}) introduced below is useful in order to partially justify the above minmax definitions as price bounds.

\begin{definition}[Relative Arbitrage] \label{relativeArbitrage}
Let $Z: \mathcal{S} \rightarrow \mathbb{R}$ be a function defined on $\mathcal{S}$ and  $\mathcal{M} = \mathcal{S} \times \mathcal{H}$ a discrete market. $H \in \mathcal{H}$, with initial value $V_0$, is a \emph{relative arbitrage} with respect to $Z$ if:
\begin{equation} \nonumber 
V_0 + \sum_{i=0}^{N_H(S)-1} H_i(S)~\Delta_iS -Z(S) \ge 0~~~\forall ~S \in \mathcal{S}\quad \mbox{and strictly positive for some}\quad S^*\in \mathcal{S}.
\end{equation}
\mbox{Or,}
\begin{equation} \nonumber 
Z(S) - V_0 + \sum_{i=0}^{N_H(S)-1} H_i(S)~\Delta_iS \ge 0~~~\forall ~S \in \mathcal{S},\quad \mbox{and strictly positive for some}\quad S^*\in \mathcal{S}.
\end{equation}
\end{definition}
Comparing with Definition \ref{ArbitrageDefinition}, $H$ with initial capital $V_0=0$ is an arbitrage strategy if and only if it is a relative arbitrage with respect to the derivative function $Z=0$.
\begin{definition}[Fair Price]
We say that $\pi$ is a \emph{fair price} for a function $Z$ in a discrete market $\mathcal{M} = \mathcal{S} \times \mathcal{H}$ if there is no $H \in \mathcal{H}$, with initial value $V_H(0, S_0)= \pi$, that is a relative arbitrage for $Z$ .
\end{definition}

It is useful to keep in mind the following obvious result.
\begin{proposition}  \label{superhedgeProposition}
Consider a discrete market $\mathcal{M} = \mathcal{S}\times \mathcal{H}$, a function $Z$ defined on $\mathcal{S}$ and $\epsilon >0$.
\begin{itemize}
\item If  $\overline{V}(S_0, Z, \mathcal{M}) > - \infty$, then there exists $H^{\epsilon} \in \mathcal{H}$: \\
$Z(S) <  \Vup(S_0, Z, \mathcal{M}) + \sum_{i=0}^{N-1} H^{\epsilon}_i(S)~\Delta_iS + \epsilon,~\mbox{for all}~~ S \in \mathcal{S}$,
where $N= N_{H^{\epsilon}}(S).$

\vspace{.1in}
\item If  $\underline{V}(S_0, Z, \mathcal{M}) < \infty$, then there exists $\hat{H}^{\epsilon}  \in \mathcal{H}$:\\
$ \Overline(S_0, Z, \mathcal{M}) + \sum_{i=0}^{N-1} \hat{H}^{\epsilon}_i(S)~\Delta_iS - \epsilon < Z(S),~\mbox{for all}~~ S \in \mathcal{S}$,
where $N= N_{\hat{H}^{\epsilon}}(S).$
\end{itemize}
\end{proposition}

Observe that by Proposition \ref{superhedgeProposition}, neither $\pi > \Vup(S_0, Z, \mathcal{M})$ nor $\pi < \Vdo(S_0, Z, \mathcal{M})$,  is a fair price  for $Z$. In the next section we are going to show conditions under which the fair prices are confined to an interval, as it is known for stochastic models. The following simple result shows that the upper-hedging and lower-hedging results in Proposition \ref{superhedgeProposition} are tight in a trajectory based sense.
\begin{proposition} \label{tightSuperhedgeAndUnderhedge}
Assume a function $Z$ defined on a discrete market model $\mathcal{M}$ is given.
\begin{itemize}
\item If $H \in \mathcal{H}$ satisfies
 $ V_{H}(0, S_0) <  \overline{V}(S_0, Z, \mathcal{M})$, then:\\
$V_{H}(N_H(S^{\ast}), S^{\ast})= V_{H}(0, S_0) +  \sum_{i=0}^{N_H(S^{\ast})-1} H_i~\Delta_iS^{\ast} < Z(S^{\ast}),~\mbox{for some}~ S^{\ast} \in \mathcal{S}.$

\vspace{.1in}
\item If $H \in \mathcal{H}$ satisfies
$ V_{H}(0, S_0) >  \Overline(S_0, Z, \mathcal{M})$, then:\\
$V_{H}(N_H(S^{\sharp}), S^{\sharp})= V_{H}(0, S_0) +  \sum_{i=0}^{N_H(S^{\sharp})-1} H_i~(S^{\sharp})~\Delta_iS^{\sharp} > Z(S^{\sharp}),~\mbox{for some}~ S^{\sharp} \in \mathcal{S}.$
\end{itemize}
\end{proposition}

\vspace{.1in}
The following proposition shows that in a general discrete market
the quantities $\overline{V}(S_0, Z, \mathcal{M}), \underline{V}(S_0, Z, \mathcal{M})$ may behave in an unexpected way (but not so in a $0$-neutral market).
\begin{proposition} \label{noAssumptionsProperty}
Given a discrete market model $\mathcal{M}= \mathcal{S} \times \mathcal{H}$ and $c$ an arbitrary constant, it follows that:
\begin{equation} \nonumber
 \mbox{if} ~~~Z(S) = c~\mbox{for all}~~S \in \mathcal{S}, ~~\overline{V}(S_0, Z, \mathcal{M}) \leq c \leq \underline{V}(S_0, Z, \mathcal{M}).
\end{equation}
In contrast, notice that if $\mathcal{M}$ is $0$-neutral then,
$\overline{V}(S_0, Z, \mathcal{M}) = c = \underline{V}(S_0, Z, \mathcal{M})$.
\end{proposition}
\begin{proof}
Consider first the case that $\mathcal{M}$ is not $0$-neutral; it follows that  $\overline{V}(S_0, Z=0, \mathcal{M}) < 0$ in that case, this implies that
there exists  $H \in \mathcal{H}$ such that:
\begin{equation} \nonumber
-\sum_{i=0}^{N_H(S)-1} H_i~(S)~\Delta_iS < 0,~\mbox{for all}~ S \in \mathcal{S},
~~\mbox{so}~~
 \sup_{S \in \mathcal{S}} [c - \sum_{i=0}^{N_H(S)-1} H_i~(S)~\Delta_iS] \le c,
\end{equation}
which leads to $\overline{V}(S_0, Z=c, \mathcal{M}) \leq c$.

\noindent Consider now that $\mathcal{M}$ is $0$-neutral;
it is then clear that $\Vup(S_0, Z=c, \mathcal{M}) =c=\Overline(S_0, Z=c, \mathcal{M}).$\\
Hence $\overline{V}(S_0, Z=c, \mathcal{M}) \leq c$ in all cases,
then $\underline{V}(S_0, Z=c, \mathcal{M}) = - \overline{V}(S_0, -Z=-c, \mathcal{M}) \geq c.$
\end{proof}

\subsection{Minmax Functions. Conditions for Boundedness of  $\underline{V}(Z)$ and $\overline{V}(Z)$}  \label{sec:minMaxFunctions}

The integrability conditions, required for payoffs in a
probabilistic setting, are replaced in the proposed framework
by the, so called, minmax functions.
The general setting  works with a general function $Z: \mathcal{S} \rightarrow \mathbb{R}$.


$Z$ is called an {\it European option} when there exists an integer $M$ and $\hat{Z}: \mathbb{R}^M \rightarrow \mathbb{R}$ and stopping
times $\tau_i \leq \tau_{i+1}$, $i=1, \ldots, M$, so that
$Z(S) = \hat{Z}(S_{\tau_1(S)}, \ldots, S_{\tau_M(S)})$.
The function $\hat{Z}$ will be called a {\it payoff}; the setting allows for path dependency.
For a European call or put option (and so $M=1$) portfolios in $\mathcal{H}$ could/should be required to satisfy
$N_H(S) \leq \tau_1(S)$ for all $H \in \mathcal{H}$ and for all $S$.

\begin{definition}[Upper and Lower Minmax Functions]\label{minmaxFunctions}
Given a finite sequence of stopping times $(\nu_i)_{i=1}^{n}$ with
$\nu_i < \nu_{i+1}$ for $1 \le i < n$, a  real sequence
$(a_i)_{i=1}^{n}$, and $b\in \mathbb{R}$, we call $Z$ an
\emph{upper minmax function} if
\begin{equation}  \nonumber
Z(S) \leq \sum_{i=1}^{n}a_i~S_{\nu_i(S)} + b,~\forall S \in
\mathcal{S}.
\end{equation}
Similarly, $Z$ is called a \emph{lower minmax function} if
\begin{equation} \nonumber
Z(S) \geq \sum_{i=1}^{n}a_i~S_{\nu_i(S)} + b,~\forall S \in
\mathcal{S}.
\end{equation}
\end{definition}

The following examples show that some common options belong to
the class of minmax functions.


\vspace{.1in}
\noindent {\bf Examples:}

\vspace{.1in}
\noindent
{\rm (1)}~ If $Z$ is an European call option with
strike price $K>0$ and  $N(S)$  a stopping time,
\[ Z(S)=(S_{N(S)}-K)^+ \le S_{N(S)}, \]
then $Z$ is an upper minmax function with $n=1$, $a_1=1$,
$\nu_1(S)=N(S)$, and $b=0$.

\vspace{.1in}
\noindent
{\rm (2)}~  If $Z$ is an European put option with
strike price $K>0$ and $N(S)$ a stopping time,
\[ Z(S)=(K-S_{N(S)})^+ \le K, \]
then $Z$ is an upper minmax function with $n=1$, $a_1=0$ and
$b=K$.

Clearly, the above two examples are also lower minmax functions.

\vspace{.1in}
\noindent
{\rm (3)}~ Under the assumption $S_k \geq 0$ for all
$S\in \mathcal{S}$ and all $k \geq 0$; if
 \[Z(S)=a~\max_{1\le i\le
n}~S_{\nu_i(S)} + b \quad \mbox{with} \quad a>0, ~~\mbox{then}~~~
Z(S) \le \sum_{i=1}^n a~S_{\nu_i(S)} + b\]
and so, $Z$ is an upper minmax function with $a_i=a$ for all
$i=1,\dots,n$.

\vspace{.1in}
\noindent
{\rm (4)}~ If
\[ Z(S)=\frac{1}{n}\sum_{i=1}^n S_{\nu_i(S)}, \]
then $Z$ is an upper minmax function with $a_i=\frac{1}{n}$ for
all $i=1, \dots, n$ and $b=0$.

Notice that, in particular,  if $S_j$ is uniformly bounded from below by a constant, for all $j$, then examples $(3)$ and $(4)$  are lower minmax functions as well.

\begin{remark} Under some assumptions on the market $\Me$, such as conditionally $0$-neutral, it can be proven that the conditional bounds $\Vup_k(S, Z, \mathcal{M})$ and/or $\Vdo_k(S, Z, \mathcal{M})$ are finite when $Z$ is an upper or lower minimax function, for reasons of space details are provided elsewhere (\cite{degano}).
\end{remark}

\section{Pricing with Arbitrage in $0$-Neutral Markets}  \label{sec:0-neutralII}

We provide general conditions resulting in a worst case price interval for the possible prices for an European option.
The notion of conditionally $0$-neutral market is the essential ingredient for the result to hold. We compare the minmax bounds  with Merton's bounds and give a detailed justification for the quantities introduced to be actual market prices. As already indicated, $0$-neutrality is a weakening of the no arbitrage condition and indeed the price interval exists even when there is a certain kind of arbitrage opportunity in the market (see discussion after Corollary \ref{necessaryCondition} in Section \ref{sec:localCharacterizations}).

\vspace{.1in}
The definition of addition of two portfolios, implicitly required in the next theorem, is introduced just before the statement of Lemma \ref{lemmaToHaveInterval} in Appendix \ref{furtherPricingResults}.

\begin{theorem} \label{havingAnIntervalTheorem}
Consider a discrete market $\Me= \Se\times \He$, a function $Z$ defined on $\Se$, $S \in \Se$ and  $k \geq 0$ fixed. Assume either that  $N_H$ is  a stopping time for all $H\in\He$ or all $H \in \mathcal{H}$ are liquidated. If $\Se\times(\He+\He)$ is conditionally $0$-neutral at node $(S, k)$, then
\begin{equation} \label{finallyTheInterval}
\underline{V}_k(S, Z, \mathcal{M}) \le \Vup_k(S, Z, \mathcal{M}),
\end{equation}
in particular,
\begin{equation}\nonumber
\underline{V}(S_0, Z, \mathcal{M}) \le \Vup(S_0,  Z, \mathcal{M}).
\end{equation}
\end{theorem}
\begin{proof}
Taking $\He^1=\He^2=\He$, the result follows directly from the conclusion (\ref{generalInterval}) of Lemma \ref{lemmaToHaveInterval} in Appendix \ref{furtherPricingResults}.
\end{proof}

Notice that assuming $\Se \times (\He + \He)$ to be conditionally $0$-neutral at node $(S, k)$ implies $\mathcal{M}$ to be conditionally $0$-neutral at that node as well. Assuming the stronger hypothesis $\mathcal{H} + \mathcal{H} = \mathcal{H}$ is not necessary, as it is clearly shown by Corollaries \ref{localHavingAnInterval} and \ref{localHavingAnIntervalII}  in
Section \ref{sec:localCharacterizations}, which provide assumptions implying the conditional $0$-neutral property; those conditions will also imply (\ref{finallyTheInterval}) and do not require
that $\mathcal{H}$ is closed under addition. \\

\vspace{0.05in}
The following is another condition on $\mathcal{S}$ that also ensures (\ref{finallyTheInterval}); the proof is immediate and so omitted.
\begin{proposition}  \label{constantTrajectoryInterval}
Consider a discrete market $\mathcal{M}= \mathcal{S}\times
\mathcal{H}$, a function $Z$  defined on $\mathcal{S}$, a fixed $S \in \mathcal{S}$ and $k \geq 0$. If there exists a sequence $S^0 \in \mathcal{S}_{(S,k)}$ such that $S^0_i=S_k~$ for all $i \ge k$,  then, $\mathcal{M}$ is conditionally $0$-neutral at $(S,k)$, and
\begin{equation}  \nonumber 
\underline{V}_k(S, Z, \mathcal{M})\le Z(S^0)\le \Vup_k(S, Z, \mathcal{M}).
\end{equation}
\end{proposition}

We provide next the simple connection between the minmax bounds and Merton's bounds \cite{merton}. For a call option $C_K(x)=(x-K)_+$, with $K>0$, Merton's bounds are $C_K(S_0)$ and $S_0$.

\begin{proposition}[Merton's Bounds Comparison] \label{MertonBounds} Consider a  discrete market
$\mathcal{M}$, an integer valued function $N= N(S)$, $S \in \mathcal{S}$, and a function $Z$ defined on $\mathcal{S}$. Assume there exists $H\in \He$ such that $H_i(S)= 1$ for any $S\in\Se$ and $0 \leq i\le N_{H}(S) \equiv N(S)$.
 We obtain:

\vspace{.1in}
($a$) If $Z(S) = C_K(S_N)$ and $\mathcal{M}$ is $0$-neutral then, $C_K(S_0)\le \Vdo(S_0,Z,\mathcal{M})$.

($b$) If $Z(S)\le S_N$ for all $S\in\Se$, $\Vup(S_0,Z,\mathcal{M})\le S_0$.
\end{proposition}
\begin{proof}
Fix $S\in\Se$, $0$-neutrality implies $\Vdo(S_0,Z,\mathcal{M}) \geq 0$, so ($a$) is clearly valid if $S_0\le K$. If $S_0>K$
\[C_K(S_0)=S_0-K \le (S_N-K)_+ - (S_N-S_0) = Z(S) - \sum_{i=0}^{N-1} \Delta_iS.\]
Thus, \[C_K(S_0)\le \inf_{S\in\Se}[Z(S) - \sum_{i=0}^{N-1} \Delta_iS]\le \Vdo(S_0,Z,\mathcal{M}).\]

($b$) $S_N - S_0 = \sum_{i=0}^{N-1} \Delta_iS$, then
\[S_0 = S_N - \sum_{i=0}^{N-1} \Delta_iS \ge Z(S) - \sum_{i=0}^{N-1} \Delta_iS.\]
Consequently
\[S_0 \ge \sup_{S\in\Se}[Z(S) - \sum_{i=0}^{N-1} \Delta_iS]\ge \Vup(S_0,Z,\mathcal{M}).
\]
\end{proof}

In a situation where Proposition \ref{constantTrajectoryInterval} and Proposition \ref{MertonBounds}, item $a)$, are both applicable, we obtain the interesting result $C_K(S_0)= \Vdo(S_0,Z,\mathcal{M})$. This shows that a characteristic of the trajectory space namely, the presence of a globally constant trajectory, implies that the lower
Merton bound is attained.

\subsection{Meaning of Option Prices in  $0$-Neutral Discrete Markets} \label{sec:pricingIn0-neutralMarkets}
Having in mind the assumptions leading to the conclusion $\underline{V}(S, Z, \mathcal{M}) \le \Vup(S,  Z, \mathcal{M})$ (as in Theorem \ref{havingAnIntervalTheorem}, Proposition \ref{constantTrajectoryInterval}, Corollary \ref{localHavingAnInterval} and Corollary \ref{localHavingAnIntervalII}),
we introduce the following definition of price interval.

\begin{definition}  \label{pricingMarket}
Consider a discrete market  $\mathcal{M}= \mathcal{S} \times \mathcal{H}$
and a function $Z$ on $\mathcal{S}$. Under the assumption
that $\underline{V}(S_0, Z, \mathcal{M}) \le \Vup(S_0,  Z, \mathcal{M})$,
we will call
$[\Overline(S_0, Z, \mathcal{M}), \Vup(S_0, Z, \mathcal{M})]$ the price interval of $Z$ relative to $\mathcal{M}$.
\end{definition}
\noindent
Observe that under the referred assumptions $\pi \in (\Overline(S_0, Z, \mathcal{M}), \Vup(S_0, Z, \mathcal{M}))$ is a fair price  for $Z$.

\vspace{.1in}
The assumptions  in Theorem \ref{havingAnIntervalTheorem} guarantee a pricing interval and at the same time allow for
arbitrage in the market (see Corollary \ref{necessaryCondition} and discussion afterwards). It should be noted that the presence of arbitrage
nodes will impact the actual value of the option  bounds. Examples for the extent to which this could happen are documented in \cite{degano}.

Under the assumption that an option contract has been traded, the existence of the minmax price interval,
independently of the presence of an arbitrage strategy, is substantiated
on the need to have enough funds to match the {\it certainty} of future financial obligations. This is in contrast to an investment in an arbitrage opportunity which profits are {\it uncertain} as they may not materialize in a $0$-neutral market. In a general $0$-neutral market, an investment following an arbitrage portfolio will not guarantee enough returns under all scenarios in order to cover the obligations required by $Z$.

The simplest mathematical example illustrating such a financial situation
is given by a one-step market $\mathcal{M}$ where $N_H(S) =1$ for all $(S,H)$
and $\sup_{\tilde{S} \in \mathcal{S}}~~ (\tilde{S}_{1} - s_{0}) >    0~~\mbox{and}~
\inf_{\tilde{S} \in \mathcal{S}} ~~(\tilde{S}_{1} - s_{0}) =  0$
and that the infimum
is realized at a unique  $\hat{S} \in \mathcal{S}$.
So, the market is locally $0$-neutral; furthermore,
if $\{h \equiv H_0(S): S  \in \mathcal{S},~H \in \mathcal{H}\} = \mathbb{R}$ one can also see that
$\overline{V}(s_0, Z, \mathcal{M})= Z(\hat{S}_1)$ where
$Z(S) = Z(S_1)$ is a European call option.
A risk neutral price is not available in this case but
the minmax price provides a solution reflecting the needs of investors dealing with the option. Namely, if the option selling price is smaller than $\overline{V}(s_0, Z, \mathcal{M})$ the potential obligation $Z(S_1)$ could not be matched under all scenarios through investing on the arbitrage (as actual profits may not materialize) resulting in a shortage of funds, under some scenarios, for the seller of the option. So, it is the worst case approach, requiring coverage under all scenarios, that allows for the co-existence of arbitrage and a price interval in a $0$-neutral market.
If $0$-neutrality does not hold, it is easy to see that
the minmax optimization falls back into the arbitrage opportunity by giving
$\overline{V}(s_0, Z, \mathcal{M})= - \infty$ and the optimal
investment $h$ given by $h^{\ast} = \infty$ in the above example.

See also related arguments in \cite{karatzas} where, in a context of portfolio selection, a numeraire portfolio is shown to exist under conditions that allow for arbitrage opportunities.

\section{Trajectory Based Conditions for  $0$-neutral and arbitrage-free Markets}\label{sec:localCharacterizations}

Theorem \ref{havingAnIntervalTheorem}, in Section \ref{sec:0-neutralII}, shows the key role of conditional $0$-neutrality
 in order to obtain a worst case price interval.
The present section provides natural and general sufficient conditions that imply a discrete market $\Me=\Se\times\He$ to be conditionally $0$-neutral or arbitrage-free. Under mild assumptions, these conditions are easily seen to be necessary conditions as well.
The key ingredients are the local conditions, introduced in Definition \ref{localDefinitions}, that allow trajectories to move in a contrarian way to an arbitrary investment. There is also a need for global conditions
related to how market participants may stop their portfolio rebalances. We provide two
general financial settings leading to such global conditions, these assumptions also imply
 existence of a price interval.

\vspace{.1in}

\vspace{.1in}
\noindent
The following definition will be a main tool.

\begin{definition}[$\epsilon$-contrarian]  \label{contrarianTrajectoryDefinition}
Given $H\in\He, S\in\Se$, $\epsilon \geq 0$  and $n\ge 1$, if
\begin{equation}  \label{contrarianTrajectory}
\exists ~~ S^{n, \epsilon} \in \mathcal{S}_{(S,n)} \quad \mbox{and} \quad
 \sum_{i=n}^{N_H(S^{n, \epsilon})-1} H_i(S^{n, \epsilon}) \Delta_i S^{n, \epsilon} < \epsilon,
\end{equation}
we will say that $H$ and $S^{n, \epsilon}$ are $\epsilon$-contrarian beyond $n$.
\end{definition}
Notice that $S^{n, \epsilon} = S$ trivially  satisfies the above definition for the case $n \geq N_H(S)$ and $\epsilon > 0$.
\begin{remark}
When the portfolio $H$ is clearly understood from the context, we will
just simply say  $S^{n, \epsilon}$ is $\epsilon$-contrarian beyond $n$. 
Also, saying ``beyond $n$'' is synonymous to the fact that $S^{n, \epsilon} \in \mathcal{S}_{(S,n)}$ where, also, $S$ is understood from the context.
\end{remark}

The following two propositions, presented without proofs, are stated so that they resemble each other, some of the similarity is lost because we have decided to deal only with the notion of arbitrage (and no arbitrage) starting only at $k=0$ (i.e. we have not introduced conditional versions of these concepts).
\begin{proposition}\label{0-neutralCharacterization}
$\Me$ is conditionally $0$-neutral at $(S,k)$,  with  $S\in\Se$ and $k\ge 0$ if and only if for each $H\in\He$
and $\epsilon >0$, there exists $S^{k, \epsilon}$ which is $\epsilon$-contrarian beyond $k$. 
\end{proposition}

\begin{proposition}\label{arbitrageFreeCharacterization}
$\Me$ is arbitrage-free if and only if for each $H\in\He$ we have:
\begin{equation} \nonumber
\exists ~~ S^{0} \in \Se~~\mbox{such that}~~ H ~ \mbox{and} ~~ S^0 ~\mbox{are}~~~0\mbox{-contrarian beyond}~ n=0,\quad
\mbox{or}\quad \sum_{i=0}^{N_H(S)-1} H_i(S) \Delta_i S = 0\quad \forall S \in \Se.
\end{equation}
\end{proposition}
Clearly if $H$ and $S$ are $\epsilon'$-contrarian then, they will be also $\epsilon$-contrarian if  $\epsilon \geq \epsilon'$, this implies the following corollary, which shows that $0$-neutral is a necessary condition for a discrete market to be arbitrage-free.
\begin{corollary}\label{necessaryCondition}
If $\Me$ is arbitrage-free, then $\Me$ is $0$-neutral.
\end{corollary}
\noindent
The converse of Corollary \ref{necessaryCondition} does not hold in general.
Consider a market with $N_H=1,~~\forall ~~H \in \mathcal{H}$. If
$\sup_{S \in \Se}~~ \Delta_{0}S >   0~~\mbox{and}~ \inf_{S \in \Se} ~~\Delta_{0}S =  0$, it provides a clear arbitrage with $H_0 \equiv 1$, nonetheless the market is $0$-neutral. We have seen in the previous section that a well defined option pricing methodology is still possible.

\vspace{.1in} The local $0$-neutral property of $\mathcal{S}$ makes it possible to obtain trajectories which are {\it almost} $\epsilon$-contrarian, this is shown in Lemma \ref{fugitiveTrajectory} in Appendix \ref{contrarianAuxiliaryMaterial}. Nevertheless, this local property, which does not involve $\He$, is not enough to obtain conditions guaranteeing conditional (global) $0$-neutrality of $\mathcal{M}$. We will tackle this shortcoming by
imposing global financially-based conditions of a  general nature that, when supplementing the trajectory based local conditions, will provide the existence of contrarian trajectories and of a price interval.

\subsection{Initially Bounded $N_H$}  \label{sec:initiallyBounded}
The following definition reflects the situation of an investor who {\it decides}  conditionally on a bounded number of transactions, that he/she will stop trading after a certain fixed number of future trades.  The setting allows for unbounded $N_H$.

\begin{definition}[Initially Bounded]\label{initiallyBoundedDef}
Given a discrete  market $\Me=\Se\times\He$ and $H \in \mathcal{H}$; we will call $N_H$ \emph{initially bounded} if there exists a bounded function $\rho:\Se \rightarrow\mathbb{N}$ (which may depend on $H$) such that 
for all $S\in \Se$:
\begin{equation}  \label{initiallyBounded}
N_H ~~~~ \mbox{is bounded on}~~~~\Se_{(S, \rho(S))}.
\end{equation}
\end{definition}
If (\ref{initiallyBounded}) holds, $\tilde{S} \in \Se_{(S, \rho(S))}$ and $\rho$ is a stopping time, then $\rho(\tilde{S})=\rho(S)$. Also,
if $N_H$ is bounded, then it is actually initially bounded by taking $\rho=N_H$.

\vspace{.1in}
We are now able to provide a general setting ensuring that a discrete market $\mathcal{M}$ is $0$-neutral.
\begin{theorem}\label{0-neutral}
Given a market $\Me=\Se\times\He$, $k\ge 0$ and $S^k\in\Se$, assume that $N_H$ is initially bounded 
for any $H\in \He$ and each node $(S,j)$, with $S \in \Se_{(S^k, k)}$ and $j\ge k$, is $0$-neutral. Then, $\Me$ is conditionally $0$-neutral at $(S^k,k)$.
\end{theorem}
\begin{proof} From Proposition \ref{0-neutralCharacterization}, for any $H\in\He$ and a given $\epsilon>0$, it is enough to show the existence of an $\epsilon$-contrarian trajectory with respect to $H$, extending $S^k$ beyond $k$. From the hypothesis on $0$-neutrality of nodes and fixed $\epsilon>0$, observe that  Lemma \ref{fugitiveTrajectory} from Appendix \ref{contrarianAuxiliaryMaterial} is applicable giving a sequence of trajectories $\{S^m\}_{m\ge k}$ verifying
\begin{equation}\label{controlledSum}
S^{m}\in \Se_{(S^{m-1},m-1)}, \quad \mbox{and}\quad
\sum_{i=k}^{n-1} H_{i}(S^m) \Delta_iS^{m}< \sum_{i=k}^{n-1}\frac \epsilon{2^i} < \epsilon, \quad
\mbox{for}\quad m\ge n > k.
\end{equation}
Since $S^{k+1}$ and $H$ result $\epsilon$-contrarian beyond $k$ if $N_H(S^{k+1})\le k+1$, we only need to consider the case where $N_H(S^{k+1})> k+1$. The result then follows from Lemma \ref{contrarianTrajectoryConstruction} $(a)$ in Appendix \ref{contrarianAuxiliaryMaterial}, taking $\kappa=k+1$ in that lemma.
\end{proof}

\vspace{.1in}
The following corollary provides existence of the pricing interval
in the setting of Theorem \ref{0-neutral}. The sum of portfolios, used in the proof, is presented before Lemma \ref{lemmaToHaveInterval}
in Appendix \ref{furtherPricingResults}.

\begin{corollary} \label{localHavingAnInterval} Consider a discrete market $\Me = \Se\times \He$, a function $Z$  defined on $\Se$, $S \in \Se$ and $k \geq 0$ fixed.
Assume $N_H$ to be initially bounded for all $H\in\He$ and either  $N_H$ is  a stopping time for all $H\in\He$ or all $H \in \mathcal{H}$ are liquidated. Then,
if each node $(\tilde{S},j)$, with $\tilde{S} \in \Se_{(S, k)}$ and $j\ge k$, is $0$-neutral:
\begin{equation} \nonumber 
\Vdo_k(S, Z, \Me) \le \Vup_k(S, Z, \Me).
\end{equation}
\end{corollary}
\begin{proof}  Observe that the initially bounded property is closed under addition.
Indeed, let $\rho_1,\rho_2$, the functions required by Definition \ref{initiallyBoundedDef} for $H^1,H^2\in\He$ respectively. Then, set $H=H^1+H^2$ and $\rho\equiv \max\{\rho_1,\rho_2\}$; since $\Se_{(S,\rho(S))}\subset \Se_{(S,\rho_j(S))}\;j=1,2$ and so, $N_{H}$ is bounded in $\Se_{(S,\rho(S))}$. Therefore, Theorem \ref{0-neutral} applies implying $~~\widetilde{\Me} =\Se\times(\He+\He)$ is conditionally $0$-neutral at $(S,k)$.
It follows that
\[\Vdo_k(S, Z, \Me) \le \Vdo_k(S, Z, \widetilde{\Me}) \le \Vup_k(S, Z, \widetilde{\Me}) \le \Vup_k(S, Z, {\Me}),\]
where the innermost inequality follows from Theorem \ref{havingAnIntervalTheorem}.
\end{proof}
\begin{remark} \label{hiding0NeutralityOfTheSum}
A more basic result is concealed in the proof of the last corollary,  indeed,  under those hypothesis $\Se\times(\He+\He)$ is conditionally $0$-neutral.
\end{remark}
In order to obtain sufficient conditions implying that a market $\mathcal{M}$ is arbitrage-free, it is conceptually clearer to work with the notion of {\it local arbitrage}. That concept represents the situation when we know a trajectory and an instance where an arbitrage opportunity will arise. It also assumes the existence of a portfolio that takes advantage of the arbitrage opportunity.

\emph{ A discrete market model $\Me = \Se \times \He$ is said to have a \emph{local arbitrage} if there exist
$S \in \Se$, $H \in \He$ and $j\ge 0$ satisfying:
\begin{equation} \label{localArbitrage}
~\inf_{\tilde{S} \in \Se_{(S,j)}}~~ [H_j(S) ~\Delta_j\tilde{S}] \geq   0,\qquad \mbox{and} \quad\qquad
~\sup_{\tilde{S} \in \Se_{(S,j)}}~~ [H_j(S) ~\Delta_j\tilde{S}] >   0.
\end{equation}}
\noindent The logical negation of the conditions in (\ref{localArbitrage}) will give local sufficient conditions leading to (global) no arbitrage results:

\emph{ A discrete market  $\Me$ is said to be \emph{free of local arbitrage} if it has no local arbitrage at any node $(S,H,j)$, that is, the following holds at any node $(S, H,j)$:
\begin{equation}\label{investPositiveAndOnlyUpII}
~\inf_{\tilde{S} \in \Se_{(S,j)}}~~ [H_j(S) ~\Delta_j\tilde{S}] <  0,
\end{equation}
or
\begin{equation}\label{investNegativeAndOnlyDownII}
~\sup_{\tilde{S} \in \Se_{(S,j)}}~~ [H_j(S) ~\Delta_j\tilde{S}] \le   0.
\end{equation}
}

\begin{remark}
Notice that above we refer to $\mathcal{M}$ as being free of local arbitrage, this is in contrast to saying, in Definition \ref{localDefinitions}, that $\mathcal{S}$ is locally arbitrage-free.
The obvious relationship is spelled out in Corollary \ref{upDownCorollary} below.
\end{remark}
The above conditions and the requirement of $N_H$ being initially bounded ensure the existence of $0$-contrarian trajectories with respect to a given $H$; that is,  we have the following result.
\begin{theorem}\label{arbitrageFreeCondition} Consider a discrete market $\Me=\Se\times \He$ free of local arbitrage. If for any $H\in \He$, $N_H$ is an initially bounded stopping time, then $\Me$ is arbitrage-free.
\end{theorem}
\begin{proof}
Fix $H\in \He$, $\hat{S}\in \Se$ and $k\ge 0$.

First observe that by Lemma \ref{arbitrageTrajectory} item $(2)$ in Appendix \ref{contrarianAuxiliaryMaterial}, for any $S \in \Se_{(\hat{S}, k)}$ either
$\sum_{i=k}^{N_H(S)-1} H_{i}(S) \Delta_iS= 0$, or there exists a smallest integer $\nu(S) \ge k$ such that (\ref{investPositiveAndOnlyUpII}) holds for ($S$,H,$\nu(S)$). Consequently, under the latter scenario
\[H_j(\tilde{S})\Delta_j\tilde{S}=0,\quad  \mbox{for} \quad k\le j < \nu(S),~~ \tilde{S}\in\Se_{(S,j)}.\]

If, for some  $S\in\Se_{(\hat{S},k)}$, (\ref{investPositiveAndOnlyUpII}) does not hold or in case (\ref{investPositiveAndOnlyUpII}) holds and $N_H(S)\le \nu(S)$ is verified,
then the second condition in Proposition \ref{arbitrageFreeCharacterization} is satisfied.
Therefore, we may assume a case in which (\ref{investPositiveAndOnlyUpII}) is valid  for some node $(S^*,H,\nu(S^*))$
with $S^*\in\Se_{(\hat{S},k)}$ and $N_H(S^*)\le \nu(S^*)$.
Applying Lemma \ref{arbitrageTrajectory}, item (\ref{3}), with $S^*$ as $S^k$, we obtain a sequence of trajectories $\left(S^m\right)_{m\ge k}$, verifying
\[
S^m = S^* \quad \mbox{for}\quad k\le m \le \nu,\quad S^{m}\in \Se_{(S^{m-1},m-1)}\quad \mbox{for}\quad m > \nu,
\]
\[
\sum_{i=k}^{\nu-1} H_{i}(S^m) \Delta_iS^{m}=0,\qquad \mbox{and }\qquad \sum_{i=k}^{n-1} H_{i}(S^m) \Delta_iS^{m}< 0, \quad \mbox{for}\quad m\ge n > \nu.
\]
Since $S^{\nu+1}\in \Se_{(S^{\nu},\nu)}= \Se_{(S^*,\nu)}$, then $N_H(S^{\nu+1}) \ge \nu +1$. So by Lemma \ref{contrarianTrajectoryConstruction} $(b)$ with $\kappa \equiv \nu+1$, there exist a trajectory $S^0$ such that $S^0$ and $H$ are $0$-contarian beyond $k$ w.r.t. $S^*$ and also to $\hat{S}$, since $\Se_{(\hat{S},k)}=\Se_{(S^*,k)}$ which, by Proposition \ref{arbitrageFreeCharacterization}, concludes the proof.
\end{proof}
Observe that the proven result is actually more general than the result stated in the theorem, it proves that the market is arbitrage-free in a conditional sense, i.e.  at each node $(\hat{S},k)$. We have not needed to formally pursue this conditional notion in the paper.

\vspace{.1in} The following result provides sufficient
conditions, involving the local arbitrage-free property of $\mathcal{S}$, leading to arbitrage-free markets.

\begin{corollary}  \label{upDownCorollary} Consider a trajectory space $\mathcal{S}$ that
is locally arbitrage-free (as per Definition \ref{localDefinitions}) and $\He$ a portfolio set. Then, $\Me = \Se \times \He$
is free of local arbitrage (as per equations (\ref{investPositiveAndOnlyUpII}) and (\ref{investNegativeAndOnlyDownII})). Moreover, if $N_H$ is an initially bounded stopping time for each $H \in \He$, $\Me$ is arbitrage-free as well.
\end{corollary}
\begin{proof} Fix $H\in\He$, $S\in\Se$ and $j\ge 0$. If $H_j(S)=0$, (\ref{investNegativeAndOnlyDownII}) is clearly
verified. While if $H_j(S)\neq 0$, then (\ref{investPositiveAndOnlyUpII}) or (\ref{investNegativeAndOnlyDownII}) are satisfied, whenever either (\ref{upDownProperty}) or (\ref{flatNode}) are valid. 
 The last assertion then follows from Theorem \ref{arbitrageFreeCondition}.
\end{proof}

\subsection{Debt Limited Portfolios}
Here we introduce a second set of financially motivated hypothesis, of a general nature,  that, when combined with
the local $0$-neutral (local arbitrage-free) assumption on $\mathcal{S}$, provide conditionally $0$-neutral (arbitrage-free) markets $\mathcal{M}$.
In fact, the following theorem shows that for all practical financial purposes, as long as the number of arbitrage and flat nodes are bounded along each trajectory, the assumption of existence of contrarian trajectories is always satisfied.
The results rely on limiting the capital that a portfolio owner may be able to borrow; this condition is usually used to exclude  arbitrage opportunities created by doubling strategies (\cite{bender2}). The setting allows for unbounded $N_H$.

\vspace{.1in}
The next theorem provides another natural and general setting, besides the one given in Theorem \ref{0-neutral}, ensuring that a discrete market is conditionally $0$-neutral.

\begin{theorem}  \label{limitedCreditStopsThePortfolio}
Given a market $\Me=\Se\times\He$, $S \in \mathcal{S}$ and $n \geq 0$. Assume each node $(S',j)$, $S' \in \Se_{(S, n)}$ with $j\ge n$, is $0$-neutral.
We further  assume:

\begin{enumerate}
\item The number of arbitrage $0$-neutral  and flat nodes (as per Definition \ref{localDefinitions}) allowed in each trajectory
is bounded by an absolute constant $\hat{m}$.

Also, for  $H \in \mathcal{H}$:

\item For given $A= A(H) \geq 0$, a constant  independent of $S$ and $k$, we have:
\begin{equation}  \label{lowerBoundForPortfolioValue}
V_0 + \sum_{i=0}^{k-1} H_i(S') \Delta_i S' \geq - A,~~~~0 \leq k \leq N_H(S'),~~~~\forall S' \in \Se_{(S, n)}.
\end{equation}

\item For given $\delta$, an absolute constant, we have:
\begin{equation}  \label{someDiscretness}
\mbox{if}~~[H_i(S') ~\Delta_i S'] \neq 0~~\mbox{then}~~|H_i(S')~\Delta_i S'| \geq \delta >0,~~~~~~\forall S' \in \Se_{(S, n)},~~~~~i \geq n.
\end{equation}
\end{enumerate}
Then,  for any  $\epsilon > 0$, there exists  $S^{n, \epsilon} \in \Se_{(S,n)}$ so that $H$ and $S^{n, \epsilon}$ are $\epsilon$-contrarian beyond $n$.
In particular, if  hypothesis $(2)$ and $(3)$ above are satisfied for all $H \in \He$,
$\Me$ is conditionally $0$-neutral at $(S, n)$.
\end{theorem}
Item $(1)$ only allows a constant maximum $\hat{m}$ of arbitrage $0$-neutral and flat  nodes along each trajectory but,  those nodes, are allowed to be arbitrarily distributed along such trajectory.
\begin{proof}
It is enough to consider the case  $n < N_H(S')$ for any $S' \in \Se_{(S,n)}$, we will establish the existence of $S^{*} \in \Se_{(S,n)}$
such that $- \sum_{i=n}^{N_H(S^{*})-1} H_i(S^{*}) \Delta_i S^{*} \geq 0$, this will conclude the proof.

Let $n_0$ be the smallest integer satisfying  $n\le n_0 < N_H(S)$ and $H_{n_0}(S) \neq 0$. If such $n_0$ does not exist we take $S^{*} \equiv S$.
There are two possibilities: a) $(S, n_0)$ is an arbitrage $0$-neutral node, b) $(S, n_0)$ is an up-down node. In case $a)$, it follows from (\ref{someDiscretness}) that there exists $S^{n, 0} \in \Se_{(S,n_0)}\subset \Se_{(S,n)}$ satisfying $\Delta_{n_0} S^{n, 0} =0$ hence
\begin{equation}  \label{firstStepa}
-H_{n_0}(S^{n, 0}) \Delta_{n_0} S^{n, 0} \geq 0.
\end{equation}
In case $b)$, from the up down property, there exists $S^{n, 0} \in \Se_{(S,n)}$
such that
\begin{equation}  \label{firstStep}
-H_{n_0}(S^{n, 0}) \Delta_{n_0} S^{n, 0} \geq \delta.
\end{equation}
If $N_H(S^{n, 0}) \le n_0+1$, since $H_i(S)=0,~~~~ n\le i <n_0$,  then (\ref{firstStepa}) or (\ref{firstStep}) show that $S^{*} \equiv S^{n, 0}$ satisfies the conditions of a contrarian trajectory we are looking for. So, assume  $n_0+1 < N_H(S^{n, 0})$.



\noindent Proceeding recursively, we may assume that we have either constructed the desired trajectory or  we have at our disposal a trajectory $S^{n,k} \in \Se_{(S^{n, k-1}, n_k)}$,
satisfying  \[-H_{n_k}(S^{n,k}) \Delta_{n_k} S^{n,k} \geq 0\quad \mbox{or}\quad -H_{n_k}(S^{n,k}) \Delta_{n_k} S^{n,k} \geq \delta\]
as well as
\[-H_{i}(S^{n, k}) \Delta_{i} S^{n, k} =0 \quad \mbox{for} \quad n_{k-1} < i < n_k.\] We then look for the smallest
$n_{k+1}$ satisfying $n_k < n_{k+1} <
N_H(S^{n,k})$ and  $H_{n_{k+1}}(S^{n,k}) \neq 0$. If such $n_{k+1}$ does not exist the construction terminates by taking $S^{*} \equiv S^{n,k}$ and so concluding the proof. Otherwise, there exists $S^{n, k+1} \in \mathcal{S}_{(S^{n, k}, n_{k+1})}$, and by means of the alternatives $a)$ and $b)$, and other considerations above, we obtain that the following holds:
\begin{equation}  \nonumber
-H_{n_{k+1}}(S^{n, k+1}) \Delta_{n_{k+1}} S^{n, k+1} \geq  0
\quad \mbox{or}\quad -H_{n_{k+1}}(S^{n, k+1}) \Delta_{n_{k+1}} S^{n, k+1} \geq \delta,
\end{equation}
as well as
\begin{equation} \nonumber
-H_{i}(S^{n, k}) \Delta_{i} S^{n, k+1} =0 \quad \mbox{for} \quad n_k < i < n_{k+1}.
\end{equation}

Continuing in this way, we have the following exclusive alternatives:
$i)$ we managed to  construct the desired trajectory and, hence, the recursion terminates. $ii)$  The recursion continues indefinitely, in which case we have:
\begin{equation}  \label{diverges}
-\sum_{k=0}^{m} H_{n_k}(S^{n, k}) \Delta_{n_k} S^{n, k}  =  -\sum_{i= n}^{n_{m}} H_{i}(S^{n, m}) \Delta_{i} S^{n, m} = -\sum_{i= n_0}^{n_{m}} H_{i}(S^{n, m}) \Delta_{i} S^{n, m} \geq  [m +1 - \hat{m})] \delta,~~~~\forall m > \hat{m},
\end{equation}
where we used the fact that $H_{n_k}(S^{n,k}) \Delta_{n_k}S^{n,k} = H_{n_k}(S^{n,m}) \Delta_{n_k}S^{n,m}$
for $0\leq k \leq m$.

Let us show that (\ref{diverges})
conflicts with (\ref{lowerBoundForPortfolioValue}):
(recall $H_i(S^{n,m})=H_i(S)=0,~~~~ n\le i <n_0$)
\begin{equation} \nonumber
 V_0 +  \sum_{i=0}^{n_{m}}H_i(S^{n, m}) \Delta_i S^{n, m}
=  V_0 + \sum_{i=0}^{n-1} H_i(S) \Delta_i S+ \sum_{i=n}^{n_{m}}H_i(S^{n, m}) \Delta_i S^{n, m}
\end{equation}
\begin{equation} \label{usedAllTheCreditAndMore}
 \leq V_0 + \sum_{i=0}^{n-1} H_i(S) \Delta_i S - [m   +1-\hat{m})]~\delta  <  -A,
\end{equation}
where we obtained  the last inequality by taking $m$ sufficiently large; let us denote the smallest integer satisfying (\ref{usedAllTheCreditAndMore}) by $m^{\ast}$.
This argument just proves that we can not have  $ n_{m^{\ast}} \leq N_H(S^{n,m^{\ast}})$ as otherwise we have a contradiction with (\ref{lowerBoundForPortfolioValue}); it then follows that $n_{m^{\ast}}  > N_H(S^{n,m^{\ast}}) > n$.
To sum up: $-H_i(S^{n, m^{\ast}}) \Delta_i S^{n, m^{\ast}} \geq 0$
for all $n \leq i < N_H(S^{n, m^{\ast}})$, $S^{n, m^{\ast}} \in \mathcal{S}_{(S,n)}$ 
hence $S^{*} \equiv S^{n, m^{\ast}}$ is a contrarian trajectory that extends $S$ beyond $n$. The conditionally $0$-neutral property then follows from Proposition \ref{0-neutralCharacterization}.
\end{proof}
\begin{remark} Notice that we have established more than is required in Definition \ref{contrarianTrajectoryDefinition} as each term in (\ref{contrarianTrajectory}) has proven to be non-negative.
The hypothesis (\ref{someDiscretness}) is only needed to extract the needed information from
the up-down nodes, the fact that that hypothesis is also used for
the arbitrage $0$-neutral and flat nodes is not essential.
\end{remark}

An study of the  proof of Theorem \ref{limitedCreditStopsThePortfolio} in conjunction with Proposition \ref{arbitrageFreeCharacterization} gives the following corollary.
\begin{corollary}  \label{secondResultOnNoArbitrage}
Assume the same hypothesis as in Theorem \ref{limitedCreditStopsThePortfolio} and, furthermore, require
$\hat{m}=0$. Then, $\mathcal{M}= \mathcal{S} \times \mathcal{H}$
is arbitrage-free.
\end{corollary}

\vspace{.1in}
The following corollary provides existence of the pricing interval
in the setting of Theorem \ref{limitedCreditStopsThePortfolio}, we borrow all assumptions from that theorem but need to strengthen (\ref{someDiscretness}) so that the addition of portfolios obeys that equation as well.

\begin{corollary} \label{localHavingAnIntervalII} Consider a discrete market $\Me = \Se\times \He$, a function $Z$  defined on $\Se$, $S \in \Se$ and $n\geq 0$ fixed. Assume that all hypothesis of Theorem \ref{limitedCreditStopsThePortfolio} are satisfied, and that either  $N_H$ is  a stopping time for all $H\in\He$ or all $H \in \mathcal{H}$ are liquidated. Moreover,
we strengthen (\ref{someDiscretness}) by assuming there are absolute constants $\delta_H >0$, $\delta_S >0$:
\begin{equation}  \label{someMoreDiscretness}
H_i(S') \in \{k ~\delta_H: k \in \mathbb{Z}\}~~~\mbox{and, whenever, }~~\Delta_i S' \neq 0~~\mbox{then}~~|\Delta_i S'| \geq \delta_S.
\end{equation}

Then:
\begin{equation} \nonumber 
\Vdo_n(S, Z, \Me) \le \Vup_n(S, Z, \Me).
\end{equation}
\end{corollary}
\begin{proof} Let $\widetilde{\He}\equiv\He+\He$, we will argue that Theorem \ref{limitedCreditStopsThePortfolio}
is applicable to $\mathcal{S} \times \widetilde{\He}$ and borrow the notation used
in that theorem. 
Assumption $(2)$ in Theorem \ref{limitedCreditStopsThePortfolio} can be made to hold for $\widetilde{\He}$
by defininig $A(H^1+H^2) = A(H^1) + A(H^2)$ whenever $H^k \in \mathcal{H}$, $k=1,2$.
Also assumption $(3)$ in Theorem \ref{limitedCreditStopsThePortfolio} holds with $\delta \equiv \delta_S ~\delta_H$ for $\widetilde{\He}$ given our assumption (\ref{someMoreDiscretness}).
Therefore, 
$~~\widetilde{\Me} =\Se\times\widetilde{\He}$ is conditionally $0$-neutral at $(S,n)$. It follows that
\[\Vdo_n(S, Z, \Me) \le \Vdo_n(S, Z, \widetilde{\Me}) \le \Vup_n(S, Z, \widetilde{\Me}) \le \Vup_n(S, Z, {\Me}),\]
where the innermost inequality follows from Theorem \ref{havingAnIntervalTheorem}.
\end{proof}

\section{Attainability. Formal Martingale Properties} \label{attainability}

This section concerns the  notion of attainability, as well as a generalization of this notion and some implications. Under the assumption of attainability the minmax bounds are additive and behave much like an integration operator. We also present  results providing formal analogues of martingale properties, in particular, a trajectory based optional stopping theorem is proven. In some cases, for the sake of generality and clarity, we directly assume the existence of a worst case pricing interval, results providing such interval are given in: Theorem \ref{havingAnIntervalTheorem}, Proposition \ref{constantTrajectoryInterval}, Corollary \ref{localHavingAnInterval} and Corollary \ref{localHavingAnIntervalII}.

\begin{definition}
Given a discrete market $\mathcal{M} = \mathcal{S} \times \mathcal{H}$,
and non-negative numbers $\epsilon^{\uparrow}$, $\epsilon^{\downarrow}$.
A function $Z$ is called $\epsilon^{\uparrow}$-upward attainable if there exists $H^{{\uparrow}} \in \mathcal{H}$
and a number $V^{\uparrow}$ such that
\begin{equation} \label{uniformBoundForProfitFromAbove}
0 \leq  V^{\uparrow}+ \sum_{i=0}^{N_{H^{\uparrow}}(S)-1} H^{\uparrow}_i(S)~\Delta_i S- Z(S) \leq
\epsilon^{\uparrow}~~~\forall ~~S \in \mathcal{S}.
\end{equation}
Analogously,
$Z$ is called $\epsilon^{\downarrow}$-downward attainable if there exists $H^{{\downarrow}} \in \mathcal{H}$
and a number $V^{\downarrow}$ such that
\begin{equation} \label{uniformBoundForProfitFromBelow}
0 \leq  -V^{\downarrow}- \sum_{i=0}^{N_{H^{\downarrow}}(S)-1} H^{\downarrow}_i(S)~\Delta_i S+ Z(S) \leq
\epsilon^{\downarrow}~~~\forall ~~S \in \mathcal{S}.
\end{equation}
Finally $Z$ is called attainable if it is $0$-upward attainable, in such a case we use the notation $H^z= H^{\uparrow}$ and $V_{H^z}= V^{\uparrow}$.
Notice that $Z$ is $0$-upward attainable if and only if it is
$0$-downward attainable.
\end{definition}

The next proposition shows that the distance separating the price bounds is bounded by the maximum profits.

\begin{proposition} \label{boundingThePriceInterval}
Let $\mathcal{M}= \mathcal{S}\times
\mathcal{H}$ be a discrete market, $S^*\in\mathcal{S}$, $k\ge0$ and $Z$ a function on $\mathcal{S}$. Consider the statements:\\
a)  $Z$ is $\epsilon^{\uparrow}$-upward attainable
and $- H^{\uparrow}  \in \mathcal{H}$.\\
b)  $Z$ is $\epsilon^{\downarrow}$-downward attainable
and $- H^{\downarrow}  \in \mathcal{H}$.\\
Then, the following holds:
\begin{equation}  \label{intervalLengthBoundedByMaximumProfits}
\overline{V}_k(S^*,Z, \mathcal{M})-
\underline{V}_k(S^*,Z, \mathcal{M}) \leq \epsilon,
\end{equation}
where $\epsilon = \epsilon^{\uparrow}$ if $a)$ holds,
$\epsilon = \epsilon^{\downarrow}$ if $b)$ holds
and $\epsilon = \epsilon^{\uparrow} \wedge \epsilon^{\downarrow}$ if $a)$ and $b)$ hold.
\end{proposition}
\begin{proof}
Introduce the notation $V^{\uparrow}(k, S^{\ast}) \equiv V^{\uparrow}+ \sum_{i=0}^{k-1} H^{\uparrow}_i(S^{\ast})~\Delta_i S^{\ast}$  and
$V^{\downarrow}(k, S^{\ast}) \equiv V^{\downarrow}+ \sum_{i=0}^{k-1} H^{\downarrow}_i(S^{\ast})~\Delta_i S^{\ast}$. If $a)$ holds,
it follows from (\ref{uniformBoundForProfitFromAbove})  and $- H^{\uparrow}  \in \He$ that
\[-\Vdo_k(S^*,Z, \Me)\le \sup_{S\in\Se_{(S^*,k)}}[-Z(S)-\sum_{i=k}^{N_{H^{\uparrow}}(S)-1}-H_i^{\uparrow}(S)\Delta_iS] = -V^{\uparrow}(k, S^{\ast}) + \epsilon^{\uparrow}.\]
\begin{equation}  \nonumber
\overline{V}_k(S^*,Z, \mathcal{M}) \leq V^{\uparrow}(k, S^{\ast}) ~~\mbox{and}~~~
\underline{V}_k(S^*,Z, \mathcal{M}) \geq  V^{\uparrow}(k, S^{\ast}) - \epsilon^{\uparrow}.
\end{equation}
So (\ref{intervalLengthBoundedByMaximumProfits}) holds.

Similarly  If $b)$ holds,
it follows  from (\ref{uniformBoundForProfitFromBelow}) that
\begin{equation} \nonumber
\overline{V}_k(S^*,Z, \mathcal{M}) \leq V^{\downarrow}(k, S^{\ast}) + \epsilon^{\downarrow} ~~\mbox{and}~~~
\underline{V}_k(S^*,Z, \mathcal{M}) \geq  V^{\downarrow}(k, S^{\ast}).
\end{equation}
So (\ref{intervalLengthBoundedByMaximumProfits}) holds.
\end{proof}

\vspace{.1in}
In general, the bounds are not linear as functions of the  payoff, the following proposition presents a case where the bounds are additive.

\begin{corollary} \label{exactReconstructionAndLinearity}
Consider a discrete market $\mathcal{M}= \mathcal{S}\times \mathcal{H}$, $S^*\in\mathcal{S}$, $k\ge0$ and $Z$ a function on $\mathcal{S}$ and assume the conditions for having ~~$\Vdo_k(S^*,Z,\Me)\le \Vup_k(S^*,Z,\Me)$~~
hold.
\\
a) If $Z$ is  attainable with portfolio $H^z$ and $-H^z \in \mathcal{H}$ then:
\begin{equation}  \label{onePointIntervalForAttainable}
\overline{V}_k(S^*,Z, \mathcal{M}) =
\underline{V}_k(S^*,Z, \mathcal{M})= V_{H^z}+ \sum_{i=0}^{k-1} H^z_i(S^{\ast}) \Delta_i S^{\ast}.
\end{equation}

\noindent
b) If $Z_j$, $j=1,2$, are  attainable with portfolios $H^{z_j}$ satisfying: $-H^{z_j} \in \mathcal{H}$ and  $H^{z_1} + H^{z_2} \in \mathcal{H}$, then,
\begin{equation} \label{additiveOnePointBoundsIfAttainable}
\overline{V}_k(S^*,Z_1+Z_2, \mathcal{M}) = \overline{V}_k(S^*,Z_1, \mathcal{M})
+ \overline{V}_k(S^*,Z_2, \mathcal{M}).
\end{equation}
\end{corollary}
\begin{proof}
By assumption,
\begin{equation} \nonumber
Z(S) = V_{H^z}(k, S) + \sum_{i=k}^{N_{H^z}(S)-1} H^z_{i}(S)\Delta_i S,~~\forall~~~S \in \mathcal{S}.
\end{equation}
Where we have used the abbreviation $V_{H^z}(k, S) \equiv V_{H^z}+ \sum_{i-0}^{k-1} H^z_i(S) \Delta_i S$, with some abuse of notation (as $V_{H^z}$ may not be necessarily equal to $V_{H^z}(0,S)$).
It then follows that,
\[
\overline{V}_k(S^*,Z, \mathcal{M})\le V_{H^z}(k, S^*).
\]
Similarly, since $-H^z \in \mathcal{H}$,  $-Z\in \mathcal{Z}(\mathcal{M})$, thus
\[
\underline{V}_k(S^*,Z, \mathcal{M}) = -\overline{V}_k(S^*,-Z, \mathcal{M}) \ge V_{H^z}(k, S^*).
\]
Notice that Proposition \ref{boundingThePriceInterval} is applicable and (\ref{intervalLengthBoundedByMaximumProfits}), together with our hypothesis, gives $\underline{V}_k(S^*,Z, \mathcal{M}) = \overline{V}_k(S^*,Z, \mathcal{M})$. This equality combined with the above inequalities concludes the proof of (\ref{onePointIntervalForAttainable}).

The proof of (\ref{additiveOnePointBoundsIfAttainable}) follows from
(\ref{onePointIntervalForAttainable}) after noticing that $Z \equiv Z_1+Z_2$ is attainable and  $V_{H^z}(k, S^*) = V_{H^{z_1}}(k, S^*) + V_{H^{z_2}}(k, S^*)$.
\end{proof}

The following result expresses a consistency result, namely today's stock price is the minmax price in a $0$-neutral discrete market $\Me$. Assumptions leading to the conclusion $\underline{V}(S, Z, \mathcal{M}) \le \Vup(S,  Z, \mathcal{M})$, for any $S\in\Se$ (as in Theorem \ref{havingAnIntervalTheorem}, Proposition \ref{constantTrajectoryInterval}, Corollary \ref{localHavingAnInterval} and Corollary \ref{localHavingAnIntervalII}), will be required. Moreover, under $0$-neutrality, the minmax operator behaves much like an expectation, in particular
the sequence $\Pi = \{\Pi_k\}$ of coordinate projections $\Pi_k:\mathcal{S} \rightarrow \mathbb{R}$, $\Pi_k(S)= S_k$ behaves like a martingale with respect to this operator.
A related result will be given later, the optional sampling theorem (see Theorem \ref{optionalSamplingTheorem}).

\begin{corollary} \label{minmaxingale}
Let $\tau$ be a stopping time and $\mathcal{M} = \mathcal{S} \times
\mathcal{H}$ a discrete market. Define $G^{\tau}$ by:
\[G_i^{\tau}(S)=1\quad \mbox{for} \quad 0 \le i \le \tau(S)-1, \quad
N_{G^{\tau}}(S) = \tau(S),\quad \mbox{and}\quad V_{G^\tau}(0,S_0)=S_0.\]
Fix $S^* \in \mathcal{S}$, $k\ge 0$ and assume the conditions on $\Me$ that assure the existence of a pricing interval.\\
If $G^{\tau}$ and $-G^{\tau}$ belong to $\mathcal{H}$, then :
\begin{equation} \label{optionalSamplingIn0Neutral}
\underline{V}_k(S^*,S_{\tau},\mathcal{M}) =
\overline{V}_k(S^*,S_{\tau},\mathcal{M}) = S^*_k.
\end{equation}
Where $S_{\tau}$ denotes the function $Z$, defined on $\mathcal{S}$ by
$Z(S)=S_{\tau(S)}.$

\vspace{.1in}
\noindent
Conversely. Assume $\tau\equiv k+1$, and for all $H \in
\mathcal{H}$, $H+ G^{\tau} \in \mathcal{H}$. Then if
$~\overline{V}_k(S^*, S_{\tau}, \mathcal{M}) = S^*_k$, it follows that
$\mathcal{M}$ is conditionally $0$-neutral at $(S^*, k)$.
\end{corollary}
\begin{proof}
 Notice that $S_{\tau(S)} = S_0 + \sum_{i=0}^{\tau(S)-1} \Delta_iS$, for any $S \in \mathcal{S}$, and $G^{\tau}$ is clearly non-anticipative. So $Z = S_{\tau}$ is attainable. 
 Therefore, Corollary \ref{exactReconstructionAndLinearity} is applicable giving (\ref{optionalSamplingIn0Neutral}) since $V_{H^{\tau}}(k,S^*) = S^*_k$.

For the second statement, if $\overline{V}_k(S^*, S_{\tau}, \mathcal{M}) = S^*_k$, then
\begin{eqnarray*} 0 &=&\inf_{H \in \mathcal{H}}\{\sup_{S \in \mathcal{S}_{(S^*, k)}}[ S_{\tau(S)} - S_k -\sum_{i=k}^{N_H(S)-1} H_i(S)~\Delta_iS) ]\} = \inf_{H \in \mathcal{H}}\{\sup_{S \in \mathcal{S}_{(S^*, k)}} [ -(\sum_{i=k}^{N_H(S)-1} H_i(S)-\sum_{i=k}^{\tau(S)-1} G^{\tau}_i(S))~\Delta_iS)]\} \\
& \le & \inf_{H \in \mathcal{H}}\{\sup_{S \in \mathcal{S}_{(S^*, k)}} [ -\sum_{i=k}^{N_H(S)-1} H_i(S)~\Delta_iS)]\}.
\end{eqnarray*}
Last inequality holds, because $H+G^{\tau} \in \mathcal{H}$ for all $H \in \mathcal{H}$, and then $\mathcal{H}\subset \{H-G^{\tau}: H \in \mathcal{H}\}$. Finally, since the portfolio $0 \in \mathcal{H}$, it follows that
\[\overline{V}_k(S^*, 0, \mathcal{M}) = 0 .\]
Therefore, $\mathcal{M}$ is conditionally $0$-neutral at $(S^*, k)$.
\end{proof}

\subsection{Trajectory Based Optional Stopping Theorem}
In our setting, the analogue of a martingale process is
a trajectory set that obeys the locally arbitrage-free  property (see Definition \ref{localDefinitions}). With this analogy in mind,
we re-cast the optional sampling theorem.

\subsubsection{Stopped Trajectory Sets}

We will make use of the following notation, given a trajectory space $\mathcal{S}$ and a stopping time $\nu$ (see Definition \ref{stoppingTimeDef} set:
\begin{equation} \nonumber
\mathcal{S}^{\nu} \equiv \{S^{\nu} = \{S^{\nu}_i\}_{i=0}^{\infty}: \exists S \in \mathcal{S},~~ S^{\nu}_i = S_{\nu(S) \wedge i}~\mbox{for all}~i \geq 0\}.
\end{equation}
\begin{remark}
Making use of the notation introduced in (\ref{generalSet}),
the augmented version of the above set is given by $\mathcal{S}^{\nu} \rightarrow \mathcal{S}^{\mathcal{W}, \nu}(s_0, w_0)
\equiv \{{\bf S}^{\nu}= \{(S^{\nu}, W^{\nu}_i)\}: \exists {\bf S} \in \mathcal{S}^{\mathcal{W}}_{\infty}(s_0, w_0) , S^{\nu}_i = S_{\nu \wedge i}, W^{\nu}_i = W_{\nu \wedge i} \}.$
\end{remark}
\begin{lemma}  \label{oneToOneRelationshipForConditioningSets}
Given a trajectory space $\mathcal{S}$ and a stopping time $\nu$ defined on $\mathcal{S}$,
fix $S \in \mathcal{S}$, $n \geq 0$ and assume
$\nu(S) \geq n+1$. Then,
\begin{equation}  \nonumber
\tilde{S}^{\nu} \in \mathcal{S}^{\nu}(S^{\nu},n)~~\mbox{if and only if}~~
\tilde{S} \in \mathcal{S}(S,n).
\end{equation}
\end{lemma}
\begin{proof}
Assume $\tilde{S}^{\nu} \in \mathcal{S}^{\nu}(S^{\nu},n)$, this means that there exist $\tilde{S} \in \mathcal{S}$ satisfying
\begin{equation} \label{equalityWhenStopping}
\tilde{S}_{\nu(\tilde{S}) \wedge k} = S_{\nu(S) \wedge k}~
\mbox{for all}~0 \leq k \leq n.
\end{equation}
Suppose that $\nu(\tilde{S}) \leq  n$, then (\ref{equalityWhenStopping}) and the fact that $\nu(S) \geq n+1$ imply:
\begin{equation} \nonumber
\tilde{S}_{k} = \tilde{S}_{\nu(\tilde{S}) \wedge k} =  S_{\nu(S) \wedge k}= S_{k}~
\mbox{for all}~0 \leq k \leq \nu(\tilde{S}).
\end{equation}
Hence $\nu(S)= \nu(\tilde{S})$ which contradicts our standing assumption $\nu(S) \geq n+1$. Therefore, it follows that $\nu(\tilde{S}) \geq n+1$.
Then (\ref{equalityWhenStopping}) implies
\begin{equation} \label{equalityAtAllPoints}
\tilde{S}_{k} = S_{k}~
\mbox{for all}~0 \leq k \leq n,
\end{equation}
and so $\tilde{S} \in \mathcal{S}(S,n)$.

Conversely, assume $\tilde{S} \in \mathcal{S}(S,n)$ and hence  (\ref{equalityAtAllPoints}) holds
and implies
\begin{equation} \nonumber
\tilde{S}_{\nu(\tilde{S}) \wedge k} = S_{\nu(S) \wedge k}~
\mbox{for all}~0 \leq k \leq n,
\end{equation}
whenever $\nu(\tilde{S}) \geq n+1$ and so in this case we have established $\tilde{S}^{\nu} \in \mathcal{S}^
{\nu}(S^{\nu},n)$. It remains then to consider the case $\nu(\tilde{S}) \leq  n$,
but this is impossible as if it were true we could conclude from
(\ref{equalityAtAllPoints}) that $\nu(\tilde{S}) = \nu(S)$ which will contradict the standing assumption that $\nu(S) \geq n+1$.
\end{proof}

\begin{lemma}  \label{keyIdentity}
Given a trajectory space $\mathcal{S}$ and a stopping time $\nu$ defined on $\mathcal{S}$,
fix $S \in \mathcal{S}$, $n \geq 0$. The following holds
\begin{equation}  \label{itWillBeNiceToUnderstandThisIdentity}
\hat{S}^{\nu}_{n+1} - S^{\nu}_n = (\hat{S}_{n+1} - S_n) ~{\bf 1}_{\{\nu(S) \geq n+1\}}(S) ~
\mbox{for all}~ \hat{S}^{\nu} \in \mathcal{S}^{\nu}(S^{\nu}, n).
\end{equation}
\end{lemma}
\begin{proof}
Consider $\hat{S}^{\nu} \in \mathcal{S}^{\nu}(S^{\nu}, n)$, so
\begin{equation} \label{equalityWhenStoppingII}
\hat{S}_{\nu(\hat{S}) \wedge k} = S_{\nu(S) \wedge k}~
\mbox{for all}~0 \leq k \leq n.
\end{equation}
We split the proof in two cases.
Case I: consider that  $\nu(S) \leq n$, hence the right hand side of (\ref{itWillBeNiceToUnderstandThisIdentity}) equals $0$. Notice that if  $\nu(\hat{S}) \geq n+1$,
then (\ref{equalityWhenStoppingII}) implies $\hat{S}_{\nu(\hat{S}) \wedge k} = S_k$
for all $0 \leq k \leq \nu(S)$ but $\nu(\hat{S}) \geq n+1 > \nu(S)$ hence
$\hat{S}_{k} = S_k$
for all $0 \leq k \leq \nu(S)$ hence $\nu(S) = \nu(\hat{S})$ which gives a contradiction. Therefore, under current Case I, we should have $\nu(\hat{S}) \leq n$ so if $p \equiv \nu(S)
\wedge \nu(\hat{S}) \leq n$ it follows from  (\ref{equalityWhenStoppingII}) that
\begin{equation} \nonumber
\hat{S}_{k} = S_{k}~
\mbox{for all}~0 \leq k \leq p,
\end{equation}
which implies, in either  case $p = \nu(S)$ or $p=\nu(\hat{S})$, that $\nu(S) = \nu(\hat{S})$. It follows that  the left hand side of (\ref{itWillBeNiceToUnderstandThisIdentity}) equals:
\begin{equation} \nonumber
\hat{S}_{\nu(\hat{S}) \wedge n+1} - S_{\nu(S)}= \hat{S}_{\nu(\hat{S})} - S_{\nu(S)}=0.
\end{equation}
Case II: consider now  $\nu(S) \geq n+1$, in this case the right hand side of
 (\ref{itWillBeNiceToUnderstandThisIdentity}) equals $(\hat{S}_{n+1} - S_n)$ which we will prove also equals its left hand side. If $\nu(\hat{S}) \leq n$ and given that $\nu(S) \geq n$ it follows from (\ref{equalityWhenStoppingII}) that $\hat{S}_k = S_k$ for all $0 \leq k \leq \nu(\hat{S})$ and this implies $\nu(S) = \nu(\hat{S})$ leading to a contradiction. Therefore, under current Case II, we should have $\nu(\hat{S}) \geq n+1$, but then
$$
\hat{S}^{\nu}_{n+1} - S^{\nu}_n =
\hat{S}_{\nu(\hat{S}) \wedge n+1} - S_{\nu(S) \wedge n} = (\hat{S}_{n+1} - S_n).
$$
\end{proof}

The following is our version of the optional stopping theorem for
martingales.
\begin{theorem}  \label{optionalSamplingTheorem}
Let $\mathcal{S}$ be a trajectory space that satisfies the locally arbitrage-free property  (locally $0$-neutral) from Definition \ref{localDefinitions} 
and $\nu$ a stopping time defined on $\mathcal{S}$. Then, $\mathcal{S}^{\nu}$
satisfies the locally arbitrage-free (locally $0$-neutral) property as well.
\end{theorem}
\begin{proof}
Fix $S^{\nu} \in \mathcal{S}^{\nu}$ and $n \geq 0$. We consider first the case when
$\nu(S) \leq n$, where $S \in \mathcal{S}$ satisfies $S^{\nu}_i = S_{\nu_i(S)}$ for all $i \geq 0$.
Then by Lemma  \ref{keyIdentity},
$(\hat{S}^{\nu}_{n+1} - S^{\nu}_n) =0$
for all~ $\hat{S}^{\nu} \in \mathcal{S}^{\nu}(S^{\nu}, n)$, therefore
\begin{equation} \nonumber
\sup_{\hat{S}^{\nu} \in \mathcal{S}^{\nu}(S^{\nu}, n)} (\hat{S}^{\nu}_{n+1} - S^{\nu}_n)
= \inf_{\hat{S}^{\nu} \in \mathcal{S}^{\nu}(S^{\nu}, n)} (\hat{S}^{\nu}_{n+1} - S^{\nu}_n)=0.
\end{equation}
In the case that $\nu(S) \geq n+1$, again by  Lemma  \ref{keyIdentity},
$(\hat{S}^{\nu}_{n+1} - S^{\nu}_n) = (\hat{S}_{n+1} - S_n)$
for all~ $\hat{S}^{\nu} \in \mathcal{S}^{\nu}(S^{\nu}, n)$,

\begin{equation}  \label{equalityForStoppedSup}
\sup_{\hat{S}^{\nu} \in \mathcal{S}^{\nu}(S^{\nu}, n)} (\hat{S}^{\nu}_{n+1} - S^{\nu}_n)
= \sup_{\hat{S}^{\nu} \in \mathcal{S}^{\nu}(S^{\nu}, n)} (\hat{S}_{n+1} - S_n)= \sup_{\hat{S} \in \mathcal{S}(S, n)} (\hat{S}_{n+1} - S_n),
\end{equation}
where we used Lemma \ref{oneToOneRelationshipForConditioningSets}. A similar result to
equation  (\ref{equalityForStoppedSup}) can be obtained for the infimum as well.
These results and the fact that $\mathcal{S}$ satisfies the locally arbitrage-free (locally $0$-neutral) property concludes the proof.
\end{proof}
Usually, the optional stopping theorem involves the statement $\mathbb{E}(X_{\tau}) = \mathbb{E}(X_{0})$,
where $X_k$ is a martingale and $\tau$ a filtration based stopping time. We have already established the analogue of this statement  under the more general hypothesis of $0$-neutrality see (\ref{optionalSamplingIn0Neutral}).

\section{arbitrage-free Markets with Non Martingale Trajectory Sets}  \label{sec:arbitrageFreeMarkets}

 In this short section we follow the work of Cheridito \cite{cheridito} and  impose a natural constraint on the portfolio set $\mathcal{H}$ that allows to provide arbitrage-free markets $\mathcal{S} \times \mathcal{H}$ but $\mathcal{S}$ being such that  can not be the support set of any martingale process.

Clearly, for finite sets  the locally arbitrage-free  property of $\mathcal{S}$ reflects
a basic property of the paths of discrete time martingales. We elaborate
more on this connection in Section \ref{relationToMartingales}.
Possible examples of discrete markets $\mathcal{M}$
which are arbitrage-free but such that $\mathcal{S}$ is not the support set of a martingale process require some additional structure. For example, if transaction costs are incurred each
time $H_j(S)$ is re-balanced, then it is possible to show that $\mathcal{M}$ is arbitrage-free while at the same time contains arbitrage nodes and so
$\mathcal{S}$ can not be recovered, in general, as paths of a martingale process (\cite{ferrando}).

One can also impose some
natural constraints on $\mathcal{H}$ so that $\mathcal{M}$ is still arbitrage-free for cases where $\mathcal{S}$ is not related to a martingale process.
Towards this goal, we present another no arbitrage result which will allow us to present examples of trajectory classes that are not the support set of a martingale process while, at the same time, the market being arbitrage-free.

The  class of examples to be introduced are motivated by Cheridito's result in a continuous-time setting (\cite{cheridito}) where  a constraint is imposed so that transactions can not be performed consecutively if  the time interval between them is smaller than an a-priori given real number
$\delta >0$. Under this constraint, fractional Brownian motion can be proven to be arbitrage-free (see also generalizations in \cite{jarrow} and \cite{bender2}).

As motivation to the formal setting below, we think that there are
continuous-time trajectories $x(t)$ and a set of times $\tau_i(x)$ interpreted as instances when a transaction occurs and, so, a new price is revealed.
There are other sets of times corresponding to each investor
who may potentially re-arrange her portfolio, these times are $\nu_i(x)$,
we require $\nu_i(x) = \tau_{j_i}$ and $ j_i < j_{i+1}$. The analogue of Cheridito constraint in this setting is to require $j_{i+1} \geq j_i +2$. The results below
incorporate this setting.

The conditions below allow to have a local upward or downward trend at some nodes as long as there is the possibility of an immediate opposite correction.
Below, we use the short hand notation $H_{-1}(S) =0$ as a convenient way to impose the hypothesis that market
participants did not have any stock holdings previous to their first trading instance at $i=0$.

\begin{definition} \label{augmentedMarket}
A discrete market  model $\Me = \Se \times \He$ is  said to allow for local fast trends if $\Se$ is local $0$-neutral and the following conditions are satisfied at all $H\in \He$, $S \in \Se$ and $j \geq 0$:
\begin{enumerate}
\item If $H_j(S) \neq H_{j-1}(S)$ then  $H_j(S) = H_{j+1}(S$).

\item For each choice of sign $\pm$:\\
if $(\sup_{\tilde{S} \in \mathcal{S}_{(S,j)}} ~\pm (\tilde{S}_{j+1} - S_{j}) = 0)$,
 then
there exists $S^{\pm} \in \mathcal{S}_{(S,j)}$ so that $(S^{\pm}_{j+1} - S_j) =0$ and
\\
$(\sup_{S' \in \Se_{(S^{\pm},j+1)}} (S'_{j+2} - S^{\pm}_{j+1}) >  0)$ and
$(\inf_{S' \in \Se_{(S^{\pm},j+1)}} ~(S'_{j+2} - S^{\pm}_{j+1}) <  0)$.
\end{enumerate}
Also assume that the stock holdings are liquidated at $N_H(S)$, and so $H_k(S)=0$ for any $k\ge N_H(S)$.
\end{definition}

We have the following result.

\begin{theorem}  \label{fastTrendsNoArbitrageInitiallyBounded}
If a discrete market $\mathcal{M} = \mathcal{S} \times \mathcal{H}$ allows for local fast trends, and $N_H$ is an initially bounded stopping time for any $H\in\He$ then it is arbitrage-free.
\end{theorem}
\begin{proof}
Fix $H\in\He$. By Proposition \ref{arbitrageFreeCharacterization}, it is enough to show that $\sum_{i=0}^{N_H(S)-1}~H_i(S)\Delta_iS< 0$ for some $S\in\Se$. 

\noindent
Let $i^{\ast}$ be the smallest integer such that there exist $S^{\ast} \in \Se$
so that $H_{i^{\ast}}(S^{\ast}) \neq 0$. Since $H_{i^{\ast}-1}(S^{\ast})=0 \neq H_{i^{\ast}}(S^{\ast})$, property (1) from Definition \ref{augmentedMarket} gives $H_{i^{\ast}+1}(S^{\ast}) \neq 0$, so $N_H(S^{\ast})>i^{\ast}+1$.

Consider first the case in which $(S^*,i^*)$ is up-down (i.e. the two inequalities in (\ref{cone}) are strict). Then, there exists $\tilde{S}\in\Se_{(S^*,i^*)}$ verifying $H_{i^*}(S^*)\Delta_{i^*}S^*<0$ and by the local $0$-neutrality of $\mathcal{S}$,  property there exists
$\hat{S}\in\Se_{(\tilde{S},i^*+1)}$ satisfying $H_{i^*+1}(\hat{S})\Delta_{i^*+1}\hat{S}\le 0$, thus
\begin{equation}\label{negativeProfit}
\sum_{i=0}^{i^*+1}~H_i(\hat{S})\Delta_i\hat{S}=\hat{c}< 0.
\end{equation}

On the other hand, if at least one inequality in (\ref{cone}) from Definition \ref{localDefinitions} is not strict, it follows that property (2) of Definition \ref{augmentedMarket} holds for
at least one of the signs $\pm$; we will handle both cases with the same argument. Therefore, there exist
$S^{\pm} \in \Se(S^{\ast}, i^{\ast})$, satisfying $(S^{\pm}_{i^{\ast}+1} - S^{\ast}_{i^{\ast}})=0$ and $\hat{S} \in \Se{S}_{(S^{\pm},i^{\ast}+1)}$ such that $H_{i^*+1}(\hat{S})\Delta_{i^*+1}\hat{S}< 0$, so (\ref{negativeProfit}) holds also with this $\hat{S}$.

In either case $H_{i^{\ast}-1}(\hat{S})= H_{i^{\ast}-1}(S^{\ast})= 0 \neq H_{i^{\ast}}(S^{\ast})= H_{i^{\ast}}(\hat{S})$, then again from property (1) from Definition \ref{augmentedMarket}, $H_{i^{\ast}+1}(\hat{S}) = H_{i^{\ast}}(\hat{S})\neq 0$, so $N_H(\hat{S})>i^{\ast}+1$ as well.

We can then apply Lemma \ref{fugitiveTrajectory} with $k\equiv i^*+2$ and $S^k\equiv \hat{S}$ and $0<\epsilon <-\hat{c}$, obtaining a sequence $(S^m)_{m\ge i^*+2}$ verifying 
\begin{equation}\nonumber
\sum_{i=i^*+2}^{n-1} H_{i}(S^m) \Delta_iS^{m}<\sum_{i=i^*+2}^{n-1}\frac \epsilon{2^i}\le \epsilon,
\quad \mbox{for}\quad n: m\ge n > i^*+2.
\end{equation}

Moreover since for all $m\ge i^*+2$, $N_H(S^m)\ge N_H(\hat{S})>i^*+1$, there exists a $0$-contrarian trajectory by an application of  Lemma \ref{contrarianTrajectoryConstruction} $(b)$ in Appendix \ref{contrarianAuxiliaryMaterial}, with $k=0$, $S^m\equiv S^*:\;0\le m\le i^*$, $S^{i^*+1}\equiv \tilde{S}$, or $S^{i^*+1}\equiv S^{\pm}$, $\epsilon = 0$ and $\kappa\equiv i^*+2>0$. Since in that case we have $S^m\in\Se_{(S^{m-1},m-1)}$ for $m>0$ and
\begin{equation}\nonumber
\sum_{i=0}^{n-1} H_{i}(S^m) \Delta_iS^{m}< \hat{c}+ \sum_{i=i^*+2}^{n-1}\frac \epsilon{2^i}\le \hat{c}+\epsilon<0,
\quad \mbox{for}\quad n: m\ge n \ge i^*+2.
\end{equation}

\end{proof}

\vspace{.1in}
\noindent
{\bf Example: Fast Local Trend Market Free of Arbitrage.}\\
In order to provide a general example of a discrete market $\mathcal{M}$ satisfying the hypothesis  of Theorem \ref{fastTrendsNoArbitrageInitiallyBounded}, consider
a trajectory set $\mathcal{S}^C$ satisfying items 2 and 3 from Definition \ref{augmentedMarket}. Let $\tau =\{\tau_i\}$ be a non-decreasing sequence
of stopping times defined on $\mathcal{S}^C$ satisfying
$\tau_0(S)=0$ for all $S \in \mathcal{S}$ also, if $\tau_{k+1}(S) > \tau_{k}(S)$,
then $\tau_{k+1}(S) \geq \tau_{k}(S)+2$. Given an arbitrary  set of non-anticipative portfolios  $\mathcal{H}$
define, for each $H \in \mathcal{H}$:
\begin{equation} \nonumber
H^C_k(S) = H_{\tau_{\hat{k}}(S)}(S),~\mbox{where}
~\hat{k}~\mbox{is the largest integer such that} ~\tau_{\hat{k}}(S) \leq k.
\end{equation}
Therefore, if $H^C_{k}(S) \neq H^C_{k-1}(S)$ it follows that $H^C_{k}(S) = H^C_{k+1}(S)$. Indeed,  $\tau_{\widehat{k-1}}(S)\le k-1<\tau_{\hat{k}}(S)\le k$, then $\tau_{\hat{k}}(S)=k$, which implies $\tau_{\hat{k}+1}(S)=k$ or $k+2$, thus $\tau_{\widehat{k+1}}(S)=k$. It also follows that the portfolios $H^C$ are non-anticipative: assume $S'_i=S_i,\;0\le i\le k$; since $\tau_{\hat{k}}(S)\le k$ it follows that
$\tau_{\hat{k}}(S')=\tau_{\hat{k}}(S)$, then $\hat{k}\le \hat{k}'$ (the largest such that $\tau_{\hat{k}'}(S') \leq k$). By symmetry $\hat{k}'\le \hat{k}$, thus $\tau_{\hat{k}'}(S')=\tau_{\hat{k}}(S)$ and
\begin{equation} \nonumber
H^C_k(S) = H_{\tau_{\hat{k}}(S)}(S)=H_{\tau_{\hat{k}}(S)}(S')=H_{\tau_{\hat{k}'}(S')}(S')=H^C_k(S').
\end{equation}
Therefore, item $(1)$  from Definition \ref{augmentedMarket}
is satisfied and hence $\mathcal{M} = \mathcal{S}^C \times \mathcal{H}^C$
satisfies the hypothesis of Theorem \ref{fastTrendsNoArbitrageInitiallyBounded} and so it is arbitrage-free.

Notice that condition $(2)$ from Definition \ref{augmentedMarket}, without imposing
conditions $(1)$ and $(3)$ as well, allows for an arbitrage opportunity. Condition
$(2)$ is  the local $0$-neutral condition introduced in previous sections.

\section{Relation to Risk Neutral Pricing}  \label{relationToMartingales}

This section defines a discrete market $\mathcal{M}$ from a continuous-time martingale market. The results give some perspective to our approach and allow to establish connections between the minmax bounds and risk neutral pricing. Trajectory spaces  are defined by stopping times samples of continuous-time martingale paths. A main point to emphasize is that the $0$-neutral property holds due to the discrete sampling via stopping times and the martingale property.

Consider a stochastic market model consisting of a probability space $(\Omega, \mathcal{F}, P)$  where
$\mathcal{F} =\{ \mathcal{F}_t\}_{0 \leq t \leq T}$ is a continuous-time filtration.
Also there is an adapted process $X = \{ X_t\}_{0 \leq t \leq T}$ taking values on $\mathbb{R}$,
we also assume  $\mathcal{F}_0$ is the trivial sigma algebra.
Moreover,  there exists
 a measure $Q$, equivalent to $P$,
such that $X$ is a martingale relative to $\mathcal{F}$ and $Q$. This setting represents an arbitrage-free (in a stochastic sense),  $1$-dimensional market with a deterministic Bank account with $0$ interest rates. An European payoff $Y$ is a real valued function defined on $\Omega$, non-negative,  $\mathcal{F}_T$-measurable with respect to $Q$.
A risk neutral price of such a claim
is then given by $\mathbb{E}_{Q}(Y)$, where the expectation is with respect to a measure $Q$.

Naturally, we assume that quantities defined on $\Omega$ are only defined a.e., we will not explicitly indicate this fact in every instance but will do so in critical aspects of the constructions. The context should make it easy to realize  if we are referring to filtration-based stopping times or trajectory-based stopping times.

\subsection{Martingale Trajectory Market:}\hspace{2in}

\vspace{.1in}
\noindent

A sequence of (filtration-based) stopping times $\tau = \{\tau_i\}$, relative to the filtration $\mathcal{F}$, is said to be admissible if $ \tau_i \leq \tau_{i+1}$, $0 = \tau_0$ and, for a given
$\omega$, there exists a smallest integer $M = M_{\tau}(\omega)$ such that $\tau_M(\omega)=T$. All sequences of stopping times considered in the remaining of this section are admissible, this fact may not be explicitly indicated.
For simplicity, we may write $X_{\tau_i(\omega)}(\omega)$, and related quantities, as $X_{\tau_i}(\omega)$.

On the stochastic side, at some points we will
look at portfolios of the form
\begin{equation} \nonumber 
u^y_0 + \sum_{i=0}^{M_{\tau}-1} U^y_{i}~ (X_{\tau_{i+1}} - X_{\tau_i}),~\mbox{for a constant}~~ u^y_0,
\end{equation}
where the investment $U^y_{i}$
is $\mathcal{F}_{\tau_i}$-measurable. For technical reasons we will assume there exists a countable subset $C$ of $[0, T]$, with $0, T \in C$,
and the quantities $U_{i}(\omega)$ depend
only on $X_s$, $s \in C$. This assumption is formalized next.

We will assume that all  $\tau= \{\tau_i\}$ are such that
the $\tau_i$ take values on $C$. Let $\Omega_0$ be a set of full measure where all random variables $X_s,~s \in C,$ are defined and let $\Omega_0(\tau)$ be a set of full measure  contained in $\Omega_0$ where all random variables $ \{X_{\tau_i}\}$
are defined.

For given $\tau$  and $\omega \in \Omega_0(\tau)$
define:
\begin{equation} \nonumber 
~x_{\omega, \tau_i}:C \rightarrow \mathbb{R}~\mbox{by}~~x_{\omega, \tau_i}(s) = X_{s \wedge \tau_i(\omega)}(\omega),
\end{equation}
also set
\begin{equation} \label{portfoliosActingOnSomeTrajectories}
\mathcal{U}(\tau) = \{U = \{U_i\}_{i\geq 0}:U_i: \Omega_0(\tau) \rightarrow \mathbb{R},
U_i(\omega) = {\bf 1}_{\{M_{\tau}> i\}}(\omega)F^U_i(x_{\omega, \tau_i})\},
\end{equation}
where
\begin{equation} \nonumber
F^U_i:\mathbb{R}^{C} \rightarrow \mathbb{R}, \mbox{ a bounded  and Borel measurable function}.
\end{equation}

Recall that the Borel subsets of $\mathbb{R}^C$, denoted by $\mathcal{B}(\mathbb{R}^C)$,  are generated by the family of cylinders
\[\{x\in \mathbb{R}^C: x(c_j)\in \Gamma, 1\le j\le n\},\]
with $\Gamma\in \mathcal{B}(\mathbb{R})$.

\begin{corollary}\label{measurabilityOfU_i}
Let $U = \{U_i\}_{i\ge 0} \in \mathcal{U}(\tau)$. Assume $M_{\tau}$ is a $\mathcal{F}_{\tau}\equiv\{\mathcal{F}_{\tau_i}\}_{i\ge 0}$ stopping time. Then
$U_i \in \mathcal{F}_{\tau_i}$ for all $i\ge 0$.
\end{corollary}
\begin{proof}
Fix $i\ge 0$. Consider the function $\phi:(\Omega,\mathcal{F}_{\tau_i}) \to (\mathbb{R}^C,\mathcal{B}(\mathbb{R}^C))$, defined by $\phi(\omega)=x_{\omega, \tau_i}$. Lemma \ref{measurableMap}, in Appendix \ref{riskNeutralPricingAuxiliaryMaterial}, shows that $\phi$ is measurable. It follows that for $F_i$ as in (\ref{portfoliosActingOnSomeTrajectories}) and any $\Gamma \in \mathcal{B}(\mathbb{R})$
\[(F_i\circ \phi)^{-1}(\Gamma) = \phi^{-1}(F_i^{-1}(\Gamma))\in \mathcal{F}_{\tau_i},\]
since $F_i^{-1}(\Gamma)\in \mathcal{B}(\mathbb{R}^C)$, thus $F_i\circ \phi$ is $\mathcal{F}_{\tau_i}$-measurable. Given that $U_i={\bf 1}_{\{M_{\tau}> i\}}(F_i\circ \phi)$, and  ${\bf 1}_{\{M_{\tau}> i\}}$ is $\mathcal{F}_{\tau_i}$-measurable, $U_i \in \mathcal{F}_{\tau_i}$ too.
\end{proof}

Regarding the trajectories defined below: we sample, a finite (but arbitrary) number of times, every random trajectory, we record those values
as well as other information that will be needed (see comments below).

Given  $\tau$, define:
\begin{equation} \label{oneSamplePathsOfMartingale}
\Swe(\tau) = \{ \mathbf{S} = (S, W)= \{(S_i, W_i)\}_{i\geq 0}: \exists ~\omega \in \Omega_0(\tau),~~~ S_i = X_{\tau_i(\omega)}(\omega), W_i= (\tau_i(\omega), x_{\omega,\tau_i})\}.
\end{equation}
Also define
\begin{equation}  \nonumber 
\Swe\equiv \cup_{\tau}\mathcal{S}^{\mathcal{W}}(\tau)~~\mbox{where the union  is taken  over admissible sequences of stopping times}.
\end{equation}

The inclusion of  $\tau_i({\omega})$ and $x_{\omega,\tau_i}$ in $W_i$ allow the functions $H_i$ and $N_H$ (defined below) to be  well defined and $H$ to be non-anticipative. Equality is defined as follows.
Let $(S, W) \in \Swe(\tau)$, $(S', W') \in \Swe(\tau')$ then: $(S, W) = (S', W')$ if and only if  $X_{\tau_i}(\omega)= X_{\tau'_i}(\omega')$,$\tau_i(\omega) = \tau'_i(\omega')$, and
$x_{\omega,\tau_i}= x_{\omega',\tau'_i},\;$ $\forall ~i \geq 0$.


As a shorthand notation, the  association, between martingale paths, the stopping time and the trajectory values,  described in (\ref{oneSamplePathsOfMartingale}) will be denoted by ${\bf S} \leftrightarrows X_{\tau}(\omega)$.

Define:
\begin{equation}  \label{martingaleTransformsAreThePortfolios}
\mathcal{H} = \{H =\{H_i\}_{i\geq 0}: H_i: \mathcal{S}^{\mathcal{W}} \rightarrow \mathbb{R}\},
\end{equation}
where the functions $H_i$ are defined as follows: there exist a bounded Borel function $F_i: \mathbb{R}^{C} \rightarrow \mathbb{R}$
with the property that, for ${\bf S} \leftrightarrows X_{\tau}(\omega)$
\begin{equation} \label{definitionOfPortfolio}
H_i({\bf S}) = {\bf 1}_{\{M_{\tau}> i\}}(\omega)F_i(x_{\omega, \tau_i}),\qquad\mbox{and}\qquad N_H({\bf S}) = M_{\tau}(\omega).
\end{equation}

\noindent
Proposition \ref{nonAnticipativePortfolios} below  shows that $H$ and $N_H$ are well defined and  $H$ is non-anticipative. We allow arbitrary values for $V_H(0, S_0)$, initial portfolio values, and define the bank account value sequence $\{B_i\}$ such that portfolios are self financing as indicated in Remark \ref{alternativeConstrOfSFPort}.

The  association described in (\ref{definitionOfPortfolio}) will be denoted by $H \leftrightarrows F$.

\begin{proposition} \label{nonAnticipativePortfolios}
Portfolios $H$ and functions $N_H$ introduced by (\ref{martingaleTransformsAreThePortfolios}) and (\ref{definitionOfPortfolio}) are well defined on $\mathcal{S}^{\mathcal{W}}$ and portfolios $H$ are non-anticipative as well.
\end{proposition}
\begin{proof} Let  ${\bf S}=(S, W),\; {\bf S}'=(S', W')\in\Swe$ with ${\bf S} \leftrightarrows X_{\tau}(\omega)$ and $ {\bf S}' \leftrightarrows X_{\tau'}(\omega')$.

\noindent 1. Assume that ${\bf S}={\bf S}'$, then $W = W'$ implies $\tau'_{j}(\omega')=\tau_{j}(\omega)~\forall j \geq 0$. Thus
\[\tau'_{j}(\omega')=\tau_{j}(\omega)<T \;\; \mbox{for}\;\; 0\le j <M_{\tau}(\omega)\quad \mbox{so} \quad  M_{\tau'}(\omega')\ge M_{\tau}(\omega).\]
Since the former reasoning is symmetric, $M_{\tau'}(\omega') = M_{\tau}(\omega)$, and $N_H$ is well defined.

Fix now $i\ge 0$, therefore by the previous statement, $M_{\tau}(\omega)> i$ if and only if $M_{\tau'}(\omega')> i$. Moreover, since $x_{\omega,\tau_i}=x_{\omega',\tau'_i}$, it follows that $H_i({\bf S}')={\bf 1}_{\{M_{\tau'}> i\}}(\omega')F_i(x_{\omega', \tau'_i})={\bf 1}_{\{M_{\tau}> i\}}(\omega)F_i(x_{\omega, \tau_i})=H_i({\bf S})$. This shows that $H$ is well defined.

\noindent 2. To prove that $H$ is no anticipative, let $i \geq 0$ fixed and assume $(S_k,W_k) = (S'_k, W'_k)$, $k=0, \ldots, i$. We need to prove that $H_i({\bf S}) = H_i({\bf S}')$. Observe that $M_{\tau}(\omega)\le i$ if and only if $M_{\tau'}(\omega')\le i$. Indeed, if $M_{\tau}(\omega)=N \le i$, then $\tau'_{N}(\omega')=\tau_{N}(\omega)=T,\;$ so $M_{\tau'}(\omega')\le N \le i$ and, consequently, $H_i({\bf S})=H_i({\bf S}')$ if that is the case.

On the other hand, since $x_{\omega,\tau_k}=x_{\omega',\tau'_k}$ for $0\le k \le i$, again
\[H_i({\bf S}')={\bf 1}_{\{M_{\tau'}> i\}}(\omega')F_i(x_{\omega', \tau'_i})={\bf 1}_{\{M_{\tau}> i\}}(\omega)F_i(x_{\omega, \tau_i})=H_i({\bf S}).\]
\end{proof}

Define the discrete {\it martingale trajectory markets}
\begin{equation}  \nonumber
\mathcal{M} = \mathcal{S}^{\mathcal{W}}~\times \mathcal{H}~~\mbox{and}~~
\mathcal{M}(\tau) = \mathcal{S}^{\mathcal{W}}(\tau)~\times \mathcal{H},
\end{equation}
where, in the case of $\mathcal{M}(\tau)$, portfolios $H$ act on $\mathcal{S}^{\mathcal{W}}(\tau) \subseteq \mathcal{S}^{\mathcal{W}}$ by restriction.

\begin{remark}
To alleviate notation, we will write $\mathcal{S}^{\mathcal{W}}_{({\bf S}, k)}(\tau)$ as $\mathcal{S}_{({\bf S}, k)}(\tau)$. Similarly, we may write $\mathcal{S}^{\mathcal{W}}_{({\bf S}, k)}$ as $\mathcal{S}_{({\bf S}, k)}$.
\end{remark}



As preparation for the next result, let $\mathcal{P}= \mathcal{P}(P) $ be the set of all martingale probability measures equivalent to $P$ and $\mathbb{E}_Q(Y)$ denotes expectation with respect to probability measure $Q$.

For the next two results in this section, we are going to assume conditions under which $\{X_{\tau_i}\}_{i\ge 0}$ behaves as a martingale with respect to ${\mathcal F}_{\tau} \equiv \{{\mathcal F}_{\tau_i}\}_{i\ge 0}$; namely:
\[
\mathbb{E}_{Q}[X_{\tau_i}|{\mathcal F}_{\tau_k}] = X_{\tau_k},\quad i\ge 1, \quad k\le i.
\]

\begin{proposition}  \label{boundingMartingalePrices}
Let $Y$ be an European payoff and, for a given $\tau$, define $Z_{\tau}({\bf S})= Y(\omega)$
where $S \leftrightarrows X_{\tau}(\omega)$ and $\omega \in \Omega_0(\tau)$. Assume that $\sup_{\bf S \in \mathcal{S}^{\mathcal{W}}(\tau)}[Z_{\tau}({\bf S}) - \sum_{i=0}^{N_{H^*}({\bf S})-1} H^*_i({\bf S}) \Delta_i S]<\infty$, for some $H^*\in\He$. Consider also that $M_{\tau}$ is a stopping time w.r.t. ${\mathcal F}_{\tau}$, and the hypothesis of Lemma \ref{martingaleDifference}, in Appendix \ref{riskNeutralPricingAuxiliaryMaterial},  are satisfied. Then,
\begin{equation} \label{ourBoundsAreLarger}
\sup_{\tau}~ \underline{V}((S_0, W_0), Z_{\tau}, \mathcal{M}(\tau))\leq \inf_{Q \in \mathcal{P}} \mathbb{E}_Q(Y)
\leq  \sup_{Q \in \mathcal{P}} \mathbb{E}_Q(Y) \leq \inf_{\tau}~   \overline{V}((S_0, W_0),
Z_{\tau}, \mathcal{M}(\tau)).
\end{equation}
\end{proposition}
\begin{proof}
Notice that $\Omega_0(\tau)$, and hence also $\mathcal{S}^{\mathcal{W}}(\tau)$, depends on $P$ only through null sets of $P$; therefore, it remains unchanged if defined through any $Q \in \mathcal{P}$.

For $i\ge 1$ set $\Omega_i(\tau)\equiv\{\omega\in \Omega_0(\tau): M_{\tau}(\omega)> i\}$, it is clear that $\Omega_i(\tau)\in \mathcal{F}_{\tau_i}$, and if $M_{\tau}$ is bounded $\Omega_0(\tau)=\Omega^c_m(\tau)$ for some $m\ge 1$. Consider $H \in \mathcal{H}$ and ${\bf S}=(S,W)\in \Swe$ with $S  \leftrightarrows  X_{\tau}(\omega)$. For any $i\ge 0$, from definition (\ref{martingaleTransformsAreThePortfolios}), $H_i({\bf S})= {\bf 1}_{\Omega_i(\tau)}(\omega) F_i(x_{\omega, \tau_i})$, with $F_i: \mathbb{R}^{C} \rightarrow \mathbb{R}$ a bounded Borel function. Defining $U_i(\omega) = {\bf 1}_{\Omega_i(\tau)}(\omega)~F_i(x_{\omega, \tau_i})$,  it follows that $U = \{U_i\}_{i\ge 0} \in \mathcal{U}(\tau)$, and  by Corollary \ref{measurabilityOfU_i} $U_i\in \mathcal{F}_{\tau_i}$. We have,
\begin{equation}\label{probabilisticRepresentation}
Y(\omega) - \sum_{i=0}^{M_{\tau}(\omega)-1} U_i(\omega) ~(X_{\tau_{i+1}}(\omega) - X_{\tau_i}(\omega)) = Z_{\tau}({\bf S}) - \sum_{i=0}^{N_H({\bf S})-1} H_i({\bf S}) \Delta_iS \le \sup_{\bf S \in \mathcal{S}^{\mathcal{W}}(\tau)}[Z_{\tau}({\bf S}) - \sum_{i=0}^{N_H({\bf S})-1} H_i({\bf S}) \Delta_i S].
\end{equation}
So, since it holds with $H^*$, and $\sum_{i=0}^{M_{\tau}-1} U_i ~(X_{\tau_{i+1}} - X_{\tau_i})$ is integrable by Lemma \ref{martingaleDifference}, $Y$ also results integrable.

Taking infimum and supremum on $\mathcal{S}^{\mathcal{W}}$ in (\ref{probabilisticRepresentation}), and then expectation w.r.t. $Q \in \mathcal{P}$, 
\begin{equation} \label{readyToTakeInfimum}
- \sup_{\bf S \in \mathcal{S}^{\mathcal{W}}(\tau)}[-Z_{\tau}({\bf S}) - \sum_{i=0}^{N_H({\bf S})-1} -H_i({\bf S}) \Delta_iS] \leq \mathbb{E}_Q[Y - \sum_{i=0}^{M_{\tau}-1} U_i ~(X_{\tau_{i+1}} - X_{\tau_i})] =
\mathbb{E}_Q[Y] \leq
\end{equation}
\begin{equation} \nonumber
\sup_{\bf S \in \mathcal{S}^{\mathcal{W}}(\tau)}[Z_{\tau}({\bf S}) - \sum_{i=0}^{N_H({\bf S})-1} H_i({\bf S}) \Delta_iS],
\end{equation}
where we have used the fact that $\{X_{\tau_i}\}_{i \geq 0}$ is a martingale \cite[Thm 1.86]{medvegyev} and Lemma \ref{martingaleDifference}.


Notice that $-  \mathcal{H}=  \mathcal{H}$, we then obtain (\ref{ourBoundsAreLarger}) by taking supremum and infimum over $H $ in the left hand side and right hand side of  (\ref{readyToTakeInfimum}) respectively and then making use of the fact that $\tau$, admissible, was taken arbitrary in the above arguments.
\end{proof}
\begin{remark} The condition $\sup_{\bf S \in \mathcal{S}^{\mathcal{W}}(\tau)}[Z_{\tau}({\bf S}) - \sum_{i=0}^{N_H({\bf S})-1} H_i({\bf S}) \Delta_i S]<\infty$ required in the previous Proposition, is equivalent to $\Vup(S_0,Z_{\tau},\mathcal{M}_{\tau})<\infty$, so it is satisfied when $Z_{\tau}$ an upper minimax function (see \cite{degano}).
\end{remark}

For the case when $\Omega$ is finite or, more generally, purely atomic, it should be clear that all nodes in $\mathcal{S}^{\mathcal{W}}$ are arbitrage-free nodes according to Definition \ref{localDefinitions} (this can readily obtained from Theorem 3.1 in \cite{taqqu}).

The next result represents the key property connecting martingale trajectory markets with
the formalism of the paper.

\begin{theorem}  \label{veryUseful}
Consider a martingale trajectory market  $\mathcal{M}(\tau)$ which satisfies the conditions of Lemma \ref{martingaleDifference} in Appendix \ref{riskNeutralPricingAuxiliaryMaterial}. Then, for
any ${\bf S} = (S, W) \in  \mathcal{S}^{\mathcal{W}}(\tau)$ and $k \geq 0$,
\begin{equation}\label{Vk 0-neutralTau}
\overline{V}_k({\bf S}, Z\equiv 0, \mathcal{M}(\tau)) = 0.
\end{equation}
That is, $\mathcal{M}(\tau)$ is conditionally $0$-neutral at any  ${\bf S}$ and
for any $k \geq 0$, according to Definition \ref{conditionally0Neutral}.
\end{theorem}
\begin{proof} Fix $\tau$ admissible, $k\geq 0$ and define, for a given $\omega \in \Omega_0(\tau)$:
\begin{equation} \nonumber
\Omega_{\omega, k}(\tau) = \{\omega' \in \Omega_0(\tau): ~X_s(\omega') = X_s(\omega)~\forall s~ \in C \cap [0, \tau_k(\omega)]~\mbox{and} ~\tau_i(\omega')= \tau_i(\omega),~\mbox{for}~0 \leq i \leq k\},
\end{equation}
Lemma \ref{measurabilityOfOmega}, in Appendix \ref{riskNeutralPricingAuxiliaryMaterial},  shows that $\Omega_{\omega, k}(\tau) \in \mathcal{F}_{\tau_k}$.

Notice that  $\omega' \in \Omega_{\omega, k}(\tau)$ implies
${\bf S}' = (S', W') \in  \mathcal{S}_{({\bf S}, k)}(\tau)$,
where $S \leftrightarrows X_{\tau}(\omega),S' \leftrightarrows X_{\tau}(\omega')$,
this claim is obvious as we have $X_{\tau'_i(\omega')}(\omega')= X_{\tau_i(\omega)}(\omega)$, $0 \leq i \leq k$.
Furthermore, $X_{s \wedge \tau'_i(\omega')}(\omega')= X_{s \wedge  \tau_i(\omega)}(\omega)~\forall ~s~ \in C \cap [0, \tau_k(\omega)]$ and so $W'_i =
W_i$, $0 \leq i \leq k$.

Given $H_i(S,W) = {\bf 1}_{\Omega_i(\tau)}(\omega)F_i(x_{\omega, \tau_i})$ where
$F_i: \mathbb{R}^{C} \rightarrow \mathbb{R}$ is bounded and Borel and $S \leftrightarrows X_{\tau}(\omega)$, define
$U_i(\omega)= {\bf 1}_{\Omega_i(\tau)}(\omega)F_i(x_{\omega, \tau_i})$ for $\omega \in \Omega_0(\tau)$.
From the above claim, the following holds everywhere on $ \Omega_{\omega, k}(\tau)$:
\begin{equation} \nonumber
[-\sum_{i=k}^{M_{\tau}(\omega)-1} U_{i}(\omega)~~(X_{\tau_{i+1}} - X_{\tau_i})(\omega)]  =
[-\sum_{i=k}^{N_H({\bf S})-1} H_i({\bf S}) ~~(S_{i+1} - S_i)] \leq 
 \sup_{\tilde{\bf S} \in \mathcal{S}_{({\bf S}, k)}(\tau)}[-\sum_{i=k}^{N_H(\tilde{\bf S})-1} H_i(\tilde{\bf S})~\Delta_i\tilde{\bf S}],
\end{equation}

\noindent Under assumptions of Lemma \ref{martingaleDifference} ($b$),
\begin{equation}  \nonumber
0 = {\bf 1}_{\Omega_{\omega, k}(\tau)} \mathbb{E}[-\sum_{i=k}^{M_{\tau}-1} U_{i}~~
(X_{\tau_{i+1}} - X_{\tau_{i}})|\mathcal{F}_{\tau_k}] \}
=  \mathbb{E}[{\bf 1}_{\Omega_{\omega, k}(\tau)}~~-\sum_{i=k}^{M_{\tau}-1} U_{i}~~
(X_{\tau_{i+1}} - X_{\tau_{i}})|\mathcal{F}_{\tau_k}]
\}~\leq
\end{equation}
\begin{equation} \nonumber
 \le {\bf 1}_{\Omega_{\omega, k}(\tau)}\sup_{\tilde{\bf S} \in \mathcal{S}_{({\bf S}, k)}(\tau)}[-\sum_{i=k}^{N_H(\tilde{\bf S})-1} H_i(\tilde{\bf S})~\Delta_i\tilde{\bf S}].
\end{equation}
Therefore
\begin{equation} \nonumber
0 \leq ~\sup_{\tilde{\bf S} \in \mathcal{S}_{({\bf S}, k)}(\tau)}[-\sum_{i=k}^{N_H(\tilde{\bf S})-1} H_i(\tilde{\bf S})~\Delta_i\tilde{\bf S}]
\end{equation}
holds for all ${\bf S}= (S, W) \in \mathcal{S}^{\mathcal{W}}(\tau)$ and so:
\begin{equation} \label{greaterOrEqualThanZero}
0 \leq \inf_{H \in \mathcal{H}} \sup_{\tilde{\bf S} \in \mathcal{S}_{({\bf S}, k)}(\tau)}[-\sum_{i=k}^{N_H(\tilde{\bf S})-1} H_i(\tilde{\bf S})~\Delta_i\tilde{\bf S}]
\end{equation}
As the portfolio $H_i =0,~~\forall ~i \geq 0$,  is in $\mathcal{H}$, it follows
from (\ref{greaterOrEqualThanZero}) that (\ref{Vk 0-neutralTau})
holds.
\end{proof}
\noindent
Theorem \ref{veryUseful} extends trivially to martingale trajectory sets of the form
$\mathcal{S}^{\mathcal{W}} = \cup_{\tau} \mathcal{S}^{\mathcal{W}}(\tau)$.

\vspace{.1in}
$Y: \Omega \rightarrow \mathbb{R}$ is called $\tau$-attainable if there exists an admissible $\tau$ and $U^y \in \mathcal{U}(\tau)$ such that
\begin{equation} \nonumber 
Y= u^y_0 + \sum_{i=0}^{M_{\tau}-1} U^y_{i}(\tau)~ (X_{\tau_{i+1}} - X_{\tau_i}),~\mbox{ a.e. for a constant}~~ u^y_0.
\end{equation}

\begin{theorem}  \label{riskNeutralMinMaxConnection2} Consider a martingale trajectory market  $\mathcal{M}(\tau)$ which satisfies the conditions of Lemma \ref{martingaleDifference} in Appendix \ref{riskNeutralPricingAuxiliaryMaterial}.
Let $Y$ be $\tau$-attainable
and define $Z({\bf S}) = Z_{\tau}({\bf S})= Y(\omega)$ where $S \leftrightarrows X_{\tau}(\omega)$ and $\omega \in \Omega_0(\tau)$. Then
\begin{equation} \nonumber
Y_k(\omega) = \overline{V}_k({\bf S}, Z, \mathcal{M}(\tau)) = \underline{V}_k({\bf S}, Z, \mathcal{M}(\tau)),~~\forall~\omega \in \Omega_0(\tau),~ ~\mbox{and}~~\forall~~0 \le k \le M_{\tau(\omega)},
\end{equation}
where $Y_k \equiv   u^y_0 +\sum_{i=0}^{k-1} U^y_{i} (X_{\tau_{i+1}} - X_{\tau_i})$, holds everywhere on $\Omega_0(\tau)$. Moreover $Y_k = \mathbb{E}(Y| \mathcal{F}_{\tau_k})$
a.e. on $\Omega$ and $\mathbb{E}(\cdot) = \mathbb{E}_{Q}(\cdot) $ and $Q$ is any martingale measure equivalent to $P$.
\end{theorem}
\begin{proof}
Notice that the following holds
for all $\omega \in \Omega_0(\tau)$,
\begin{equation}  \label{decomposition}
\sum_{i=0}^{k-1} U^y_{i}(\omega) (X_{\tau_{i+1}}(\omega) - X_{\tau_i})(\omega)) + u^y_0=
Y(\omega)-  \sum_{i=k}^{M_{\tau}(\omega)-1} U^y_{i}(\omega) (X_{\tau_{i+1}}(\omega) -
X_{\tau_{i}}(\omega)),\quad
\end{equation}
which, given that $S \leftrightarrows X_{\tau}(\omega)$, is the same as:
\begin{equation}  \nonumber 
\sum_{i=0}^{k-1} H^z_{i}({\bf S})~\Delta_iS +u_0^y= Z({\bf S})- \sum_{i=k}^{N_{H^z}({\bf S})-1} H^z_{i}({\bf S})~\Delta_iS.
\end{equation}
We have used the notation $H_i^z({\bf S}) \equiv {\bf 1}_{\Omega_i(\tau)}F_i^{U^y}(x_{\omega, \tau_i}) = U^y_i(\omega)$.

Taking the conditional expectation of both sides of (\ref{decomposition})
with respect to $\mathcal{F}_{\tau_k}$ and using Lemma \ref{martingaleDifference}, gives
\begin{equation} \label{conditioningIII}
Y_k= \sum_{i=0}^{k-1} U^y_{i} (X_{\tau_{i+1}} - X_{\tau_{i}})+ u^y_0  = \mathbb{E}(Y-
\sum_{i=k}^{M_{\tau}-1} U^y_{i} (X_{\tau_{i+1}} - X_{\tau_{i}}))|\mathcal{F}_{\tau_k})= \mathbb{E}(Y|\mathcal{F}_{\tau_k}).
\end{equation}
The right hand side of (\ref{conditioningIII}) is only defined a.e. on $\Omega$; in the case that is not
defined everywhere on $\Omega_0(\tau)$ (which, we recall, is a set of probability one) we do
extend $\mathbb{E}(Y|\mathcal{F}_{\tau_k})$ to all of $\Omega_0(\tau)$ by means of the left hand side of (\ref{conditioningIII}).

At this point we recall Corollary \ref{exactReconstructionAndLinearity}, which is applicable to $\mathcal{M}(\tau)$ because of Theorem \ref{veryUseful}, which gives:
\begin{equation}  \nonumber
\overline{V}_k({\bf S}, Z, \mathcal{M}(\tau)) = \underline{V}_k({\bf S},Z, \mathcal{M}(\tau)) =  \sum_{i=0}^{k-1} H^z_{i}({\bf S})~\Delta_iS +u_0^y= Y_k(\omega)
\end{equation}
valid for all $\omega \in \Omega_0(\tau)$ and ${\bf S}= X_{\tau}(\omega)$. 
\end{proof}

Model uncertainty is usually treated by considering a subset of
the set of equivalent measures.
There are examples of stochastic market models for which the bounds $\inf_{Q}(Y)$ and $\sup_{Q}(Y)$ provide too large of an interval in order to be informative for practical purposes.
From a trajectory point of view such a situation suggests: a) a deficiency of the
market model (in particular the trajectory set $\mathcal{S}$ may be too large)
or b) the need to replace the super-hedging philosophy (and hence risk-free approach)
for a risk taking philosophy. In this last case, the error functional used to define the bound $\overline{V}$ has to be replaced by an appropriate, trajectory based, risk-functional. There are several other logical possibilities besides $a)$ and $b)$, for example including liquid derivatives in the portfolio approximations (see \cite{acciaio}).

We have considered the set $\mathcal{P}$, it is natural to seek an  extension of the above
results to the case on non-equivalent martingale measures.

\section{Conclusions and Extensions} \label{sec:conclusions}

The paper develops basic results on arbitrage and pricing in a trajectory based market model. The setting naturally allows one to resort to a worst case point of view which, in turn, permits arbitrage opportunities while at the same time providing coherent prices.
This fact reveals a basic extension to the classical martingale market structure.
The proposed framework has also a clear conceptual and formal relationship to the well established risk-neutral approach. Given the basic nature of the arguments it is expected that extensions to other settings
are possible as well.

We have concentrated on bounding the price of an option through superhedging and underhedging, selecting an actual price inside of this interval may require to adopt a functional to accommodate the ensuing risk-taking.

Arguably, attempting a direct evaluation of the minmax optimization required in (\ref{upperBound}) and in related results, is a daunting task. Moreover, the minmax formulation of the problem gives no clues on how to construct the hedging values $H_i(S)$, for a given  payoff $Z$, by means of the unfolding path values $S_0, S_1, S_2, \ldots $
In the paper \cite{degano}, and following \cite{BJN}, we propose another pair of numbers, obtained through a dynamic, or iterative, definition, each instance involving a local minmax optimization. Using the new dynamic minmax definitions we provide conditions  under which the global and the iterated definitions coincide.

The manuscript \cite{ferrando} extends and generalizes some of the no arbitrage results to the case of transaction costs. We are also presently studying continuos-time versions of the main results as well.

\appendix


\section{Further Results on Price Bounds} \label{furtherPricingResults}


\noindent For completeness, we just state the following simple result

\begin{proposition}
Consider a discrete $0$-neutral market $\mathcal{M} = \mathcal{S} \times \mathcal{H}$ and functions $Z_1(S)$, $Z_2(S)$ satisfying
\[ Z_2(S)=aZ_1(S)+b \] where
 $a$ y $b$ are arbitrary real numbers. Then,
\begin{enumerate}
\item If $a>0$ and  $\mathcal{H}$ is closed under multiplication by positive numbers:
\[ \Vup(S_0, Z_2, \mathcal{M})=a\Vup(S_0, Z_1, \mathcal{M})+b. \]
\item If $a<0$ y $\mathcal{H}$ is closed under multiplication by positive numbers:
\[ \Vup(S_0,Z_2, \mathcal{M})=a\Overline(S_0,Z_1, \mathcal{M})+b. \]
\end{enumerate}
\label{prop:2}
\end{proposition}

The following developments  are stated and proven with two portfolio sets $\He^1,\He^2$, the reader could take $\mathcal{H}= \mathcal{H}^1= \mathcal{H}^2$, as the extra generality is not used in the rest of the paper.

Define
\[\He^1+\He^2 =\{H^1+H^2: H^1\in\He^1, ~ H^2\in\He^2\},\]
where the sum $H\equiv H^1+H^2$ is defined as follows
\[N_{H}\equiv\max\{N_{H^1},N_{H^2}\},\]
\begin{equation} \label{sumOfPortfolios}
H_i=H_i^1+H_i^2\quad \mbox{if}\quad 0\le i < \min\{N_{H^1},N_{H^2}\}\quad \mbox{and}\quad H_i=H_i^j\quad \mbox{if}\quad \min\{N_{H^1},N_{H^2}\}\le i,
\end{equation}
where $N_{H}$ is attained in $H^j$, $j=1$ or $j=2$. We now check that the portfolio sum $H$ is non-anticipative under the assumption that  $N_{H^j},\;j=1,2$, are stopping times.
Let $S'_j=S_j,\;0\le j\le i$; if $i< \min\{N_{H^1}(S),N_{H^2}(S)\}\equiv m$, it is clear that $H_i(S)=H_i(S')$.  Consider then $i\ge m$ and assume, without lost of generality, that $N_H^1(S)= m$, then $N_{H^1}(S')=N_{H^1}(S)$. If $N_{H^2}(S')<N_{H^1}(S')$, it would result that $S'_j=S_j,\;0\le j\le N_{H^2}(S')$, and so $N_{H^2}(S')=N_{H^2}(S)\ge N_{H^1}(S)=N_{H^1}(S')$, a contradiction.

In case that the functions $N_{H^j}$  are not sopping times, the portfolio sum $H$ is still non-anticipative
if liquidation is assumed. Indeed, if the portfolios $H^j$ are liquidated at $N_{H^j},\; j=1,2$ (i.e. for any $S\in\Se,$ $H^j_i(S)=0$ for $i\ge N_{H^j}(S)$), the sum definition in (\ref{sumOfPortfolios}) reduces to
\[H_i=H_i^1+H_i^2\quad \mbox{for any}\quad i\ge 0.\]
It is clear that if $S,S'\in\Se$ with $S'_j=S_j,\;0\le j\le i$, for some $i\ge0$, then $H_i(S')=H^1_i(S')+H^2_i(S')=H^1_i(S)+H^2_i(S)=H_i(S)$.

Observe also that for any $S\in \Se$ and $k\ge 0$,
\[\sum_{i=k}^{N_{H}(S)-1} H_i(S) ~\Delta_i S= \sum_{i=k}^{N_{H^1}(S)-1} H^1_i(S) ~\Delta_i S ~ + ~ \sum_{i=k}^{N_{H^2}(S)-1} H^2_i(S) ~\Delta_i S.\]

\begin{lemma} \label{lemmaToHaveInterval}  Let $\Me^1 = \Se\times \He^1$ and $\Me^2 = \Se\times \He^2$ be discrete markets, and assume either: for all $H^j\in \He^j$, $j=1,2$, $N_{H^j}$ are stopping times or, all $H^j \in \He^j$, $j=1,2$, are liquidated. Set $\widetilde{\Me}= \Se\times
(\He^1 + \He^2)$ and $S \in \Se$ and $0 \leq k$. Assume $Z_1, Z_2,Z$ are real valued functions defined on $\Se$
then:
\begin{equation}\label{SubaditivityOfVup_k}
\Vup_k(S, Z_1+Z_2, \widetilde{\Me})\le \Vup_k(S, Z_1,\Me^1)+\Vup_k(S, Z_2, \Me^2).
\end{equation}
Moreover if $\widetilde{\Me}$ is conditional $0$-neutral at $(S,k)$ then
\begin{equation}\label{generalInterval}
\Vdo_k(S, Z, \Me^1)\le \Vup_k(S, Z, \Me^2).
\end{equation}
\end{lemma}
\begin{proof}
Let $H^j$ be generic elements of $\He^j$,$\;j=1,2$, so $H^1+H^2\in \He^1 + \He^2$. Then
\[\Vup_k(S,Z_1+Z_2,\widetilde{\Me}) \leq
\sup_{\tilde{S} \in \Se_{(S,k)}} [Z_1(\tilde{S}) -  \sum_{i=k}^{N_{H^1}(\tilde{S})-1} H^1_i(\tilde{S}) ~\Delta_i \tilde{S} ~ + ~ Z_2(\tilde{S}) - \sum_{i=k}^{N_{H^2}(\tilde{S})-1} H^2_i(\tilde{S}) ~\Delta_i \tilde{S}]\le\]
\[\le \sup_{\tilde{S} \in \Se_{(S,k)}} [Z_1(\tilde{S}) -  \sum_{i=k}^{N_{H^1}(\tilde{S})-1} H^1_i(\tilde{S}) ~\Delta_i \tilde{S}] + \sup_{\tilde{S} \in \Se_{(S,k)}} [Z_2(\tilde{S}) - \sum_{i=k}^{N_{H^2}(\tilde{S})-1} H^2_i(\tilde{S}) ~\Delta_i \tilde{S}].\]
Therefore, taking infimum over $\He^1$ and $\He^2$,
\[\Vup_k(S,Z_1+Z_2,\widetilde{\Me}) \le \Vup_k(S,Z_1,\Me^1)+\Vup_k(S,Z_2,\Me^2).\]
This proves (\ref{SubaditivityOfVup_k}). For (\ref{generalInterval}), by the result on (\ref{SubaditivityOfVup_k}) with $Z_1=-Z$ and $Z_2=Z$, and the conditional $0$-neutral property of $\widetilde{\Me}$ we have
\[0 = \Vup_k(S,0,\widetilde{\Me})\le  \Vup_k(S,-Z, \mathcal{M}^1) + \Vup_k(S, Z, \mathcal{M}^2).\]
Which gives the desired result.
\end{proof}


\section{Contrarian Trajectory Auxiliary Material} \label{contrarianAuxiliaryMaterial}

\begin{lemma}  \label{fugitiveTrajectory} Given a discrete market $\Me=\Se\times
\He$, $k \geq 0$, $S^k\in\Se$, and $H\in\He$. Assume each node $(S, j)$, with
$S \in \Se_{(S^k, k)}$ and $j\ge k$, is $0$-neutral. Then, for any $\epsilon > 0$, there exists a sequence of
trajectories $\left(S^m\right)_{m\ge k}$, verifying, (\ref{controlledSum}), it is
\begin{equation}\nonumber 
S^{m}\in \Se_{(S^{m-1},m-1)}, \quad \mbox{and}\quad
\sum_{i=k}^{n-1} H_{i}(S^m) \Delta_iS^{m}< \sum_{i=k}^{n-1}\frac \epsilon{2^i}\le \epsilon, \quad
\mbox{for}\quad m\ge n > k.
\end{equation}
\end{lemma}
\begin{proof} Fix $\epsilon >0$. By the assumed $0$-neutral property of the nodes, there exists  $S^{k+1}\in
\Se_{(S^k, k)}$ such that
\begin{equation}\label{contrarianStep}
H_k(S^{k+1})\Delta_kS^{k+1} < \frac {\epsilon} {2^k}.
\end{equation} Recursively, for $m\ge k+1$, once $S^m\in
\Se_{(S^{m-1}, m-1)}$ was chosen, verifying
(\ref{contrarianStep}), with $m-1$ taking the place of $k$, and then
\begin{equation}\label{contarianCumulative}
\sum_{i=k}^{m-1}H_i(S^m)\Delta_iS^m < \sum_{i=k}^{m-1} \frac
{\epsilon} {2^i}\le \epsilon,
\end{equation}
there exists  $S^{m+1}\in \mathcal{S}_{(S^{m}, m)}$ verifying
(\ref{contrarianStep}) with $k$ replaced by $m$, and
(\ref{contarianCumulative}) with $m$ replaced by $m+1$.
\end{proof}

\begin{lemma}\label{arbitrageTrajectory} Given a discrete market $\Me=\Se\times
\He$,  $k \geq 0$, $S^k\in\Se$, and $H\in\He$. Assume $\Me$ is free of local arbitrage. Then
\begin{enumerate}
\item \label{1}If (\ref{investPositiveAndOnlyUpII}) does not hold for $(S,H,j)$, with $S \in \Se_{(S^k, k)}$, and $j\ge k$, then
\[H_j(\tilde{S})\Delta_j\tilde{S}=0\quad \mbox{for any}\quad \tilde{S}\in\Se_{(S,j)}.\]

\item \label{2} For each $S \in \Se_{(S^k, k)},$ \ \ \ $\sum_{i=k}^{N_H(S)-1} H_{i}(S) \Delta_iS= 0,~~$ \ \ \ or

there exists a first integer $\nu(S) \ge k$ such that (\ref{investPositiveAndOnlyUpII}) holds for ($S$,H,$\nu(S)$).

\item \label{3} If $(S^k,H,j)$ satisfies (\ref{investPositiveAndOnlyUpII}) for some $j\ge k$, then there exists a sequence of trajectories $\{S^m\}_{m\ge k}$, verifying $S^m=S^k$ for $k<m\le\nu\equiv\nu(S^k)$, $S^{m}\in \Se_{(S^{m-1},m-1)}$, for any $m>\nu$,
\begin{equation}\label{initialSum}
\sum_{i=k}^{\nu-1} H_{i}(S^m) \Delta_iS^{m}=0,
\end{equation}
and
\begin{equation}\label{arbitrageSum}
\sum_{i=k}^{n-1} H_{i}(S^m) \Delta_iS^{m}< 0, \qquad \mbox{for}\quad m\ge n > \nu.
\end{equation}
\end{enumerate}
\end{lemma}
\begin{proof} Under the hypothesis (\ref{1}) since (\ref{investNegativeAndOnlyDownII}) must hold, for any $\tilde{S} \in \Se_{(S,j)}$,
\[0\le \inf_{\tilde{S} \in \Se_{(S,j)}}~~ [H_j(S) ~\Delta_j\tilde{S}]\le H_j(S)\Delta_j\tilde{S} \le \sup_{\tilde{S} \in \Se_{(S,j)}}~~ [H_j(S) ~\Delta_j\tilde{S}]\le 0.\]

Item (\ref{2}). If (\ref{investPositiveAndOnlyUpII}) holds for $(S,H,j)$ for some $j\ge k$, choose $\nu(S)$ as the first integer such that (\ref{investPositiveAndOnlyUpII}) holds for $(S,H,\nu(S))$. On the other hand, by (\ref{1}), $\sum_{i=k}^{N_H(S)-1} H_{i}(S) \Delta_iS= 0$.

Finally, by (\ref{2}) and the hypothesis in (\ref{3}) there exists $S^{\nu+1}\in \Se_{(S^{k},\nu)}$ with
$H_{\nu}(S^{\nu+1})\Delta_{\nu}S^{\nu+1} < 0$. If $\nu>k+1$, define, for $k < m \le \nu$, $S^m=S^{k}$.  from the minimality of $\nu$, and (\ref{1}), $S^m$ verifies (\ref{initialSum}) for those $m$. While $S^{\nu+1}$  verifies both (\ref{initialSum}) and (\ref{arbitrageSum}).


Recursively, by the free of local arbitrage property for, $m > \nu+1$, once $S^m\in
\mathcal{S}_{(S^{m-1}, m-1)}$ was chosen verifying (\ref{arbitrageSum}),
it is possible to select $S^{m+1}\in \mathcal{S}_{(S^{m}, m)}$ verifying $H_{m}(S^{m+1}) \Delta_{m}S^{m+1} \le 0$ and consequently verifying (\ref{arbitrageSum}).
\end{proof}

Combining the $0$-neutral, or free of local arbitrage conditions, of the nodes with $N_H$ being initially bounded we have.

\begin{lemma}\label{contrarianTrajectoryConstruction} Consider a discrete market $\mathcal{M}=\mathcal{S}\times
\mathcal{H}$, $k \geq 0$, $S^k\in\Se$, $\epsilon\ge 0$ and $H\in\He$ with $N_H$ initially bounded. Assume there exist a sequence of trajectories $\{S^m\}_{m\ge k}$ verifying $S^{m}\in \Se_{(S^{m-1},m-1)}$ for $m>1$, and $\kappa > k$, such that $N_H(S^\kappa)\ge \kappa$ and
\begin{equation}\nonumber 
\sum_{i=k}^{n-1} H_{i}(S^m) \Delta_iS^{m}< \sum_{i=k}^{n-1}\frac \epsilon{2^i}\le \epsilon, \quad
\mbox{for}\quad n: m\ge n \ge \kappa.
\end{equation}

\noindent Then there exists $m^*>k$ verifying that. 

\noindent $(a)$  $H$ and $S^{m^*}$ are $\epsilon$-contarian  beyond $k$, if $\kappa = k+1$ and $\epsilon > 0$.

\noindent $(b)$  $H$ and $S^{m^*}$ are $0$-contarian  beyond $k$, if $N_H$ is a stopping time.
\end{lemma}
\begin{proof} Observe that it is enough to show the existence of such $m^*>k$, with $N_H(S^{m^*})\le {m^*}$ for $(a)$, and for $(b)$ with $\kappa \le N_H(S^{m^*})\le m^*$.

If $N_H(S^\kappa) = \kappa$, $(a)$ and $(b)$ holds with $m^*=\kappa$. Let then assume that $N_H(S^\kappa) > \kappa$. Set $\mu=\max\left(\{\kappa\}\cup\{ \rho(S):S\in\Se\}\right)$, where $\rho$ is the bounded function required by Definition \ref{initiallyBoundedDef}. 

If $N_H(S^{\mu})\le \mu$, $(a)$ holds with $m^*=\mu$. For $(b)$, since $S^{\mu}\in \Se_{(S^{\kappa},\kappa)}$ because $\mu \ge \kappa$, then $N_H(S^{\mu}) \ge \kappa$, using that $N_H$ is a stopping time and $N_H(S^{\kappa}) \ge \kappa$, so $m^*=\mu$ is also a desired integer.

Consider now the case that $N_H(S^{\mu})> \mu$, and $m^*\equiv\max\{ N_H(S):S\in\Se_{(S^{\mu},\rho(S^{\mu}))}\}$. Then from $\rho(S^{\mu})\le \mu \le N_H(S^{\mu})\le m^*$ results $S^{m^*}\in \Se_{(S^{\mu},\mu)}\subset \Se_{(S^{\mu},\rho(S^{\mu}))}$ and so $N_H(S^{m^*})\le m^*$. Then the conclusion $(a)$ is verified. Item $(b)$ is also valid with that $m^*$, since $\mu \ge \kappa$, then $S^{m^*}\in \Se_{(S^{\kappa},\kappa)}$ and using that $N_H$ is a stopping time, $N_H(S^{m^*})\ge \kappa$.

\end{proof}

\section{Connections with Risk Neutral Pricing.  Auxiliary Material} \label{riskNeutralPricingAuxiliaryMaterial}

\begin{lemma}   \label{measurableMap}
The function $\phi:(\Omega,\mathcal{F}_{\tau_i}) \to (\mathbb{R}^C,\mathcal{B}(\mathbb{R}^C))$ defined by $\phi(\omega)=x_{\omega, \tau_i}$ is measurable.
\end{lemma}
\begin{proof}
For $1\le j\le n$, fix $c_j\in C, \Gamma_j\in \mathcal{B}(\mathbb{R})$ and let $\mathcal{C}=\{x\in \mathbb{R}^C: x(c_j)\in \Gamma_j, 1\le j\le n\}$. Thus
\[\phi^{-1}(\mathcal{C})=\cap_{j=1}^n\{\omega:x_{\omega, \tau_i}(c_j)\in\Gamma_j\}= \cap_{j=1}^n X^{-1}_{c_j \wedge \tau_i}(\Gamma_j). \]
For showing that $\phi^{-1}(\mathcal{C})\in \mathcal{F}_{\tau_i}$, it is then enough to prove that, for any $c\in C$ and $\Gamma\in \mathcal{B}(\mathbb{R})$, $X^{-1}_{c \wedge \tau_i}(\Gamma)\in\mathcal{F}_{\tau_i}$. This happens, if for any $t\ge 0$
\[A = \{\omega: X_{c \wedge \tau_i(\omega)}(\omega)\in\Gamma,\quad \tau_i(\omega)\le t\} \in \mathcal{F}_t.\]
To prove this, lets first define $B=\{\omega: X_{\tau_i(\omega)}(\omega)\in\Gamma,\quad \tau_i(\omega)\le t\}$ and consider two cases.

I. Assume $t\le c$, then for $\omega\in A$, $c \wedge \tau_i(\omega)=\tau_i(\omega)$ which implies that $A\subset B$. Conversely, if $\omega\in B$ then $c \wedge \tau_i(\omega)=\tau_i(\omega)$, and $B\subset A$. Now we are going to prove that
\[B=\cup_{\{s\in C:s\le t\}}\{\omega:\tau_i(\omega)=s, X_s(\omega)\in\Gamma\}.\]
Indeed, if $\omega\in B$ there exists $s\in C$ such that $s=\tau_i(\omega)\le t$ and then  $X_{\tau_i(\omega)}(\omega)=X_s(\omega)\in\Gamma$. The converse is also clearly true. Finally, since for each $s\in C, s\le t$; $\{\omega:\tau_i(\omega)=s, X_s(\omega)\in\Gamma\}\in \mathcal{F}_s \subset \mathcal{F}_t$ , it follows that $A\in \mathcal{F}_t$.

II. The case when $c< t$ follows from the decomposition of $A$, as
\[A=(\{\omega: \tau_i(\omega)\le c\}\cap B) \cup(\{\omega: \tau_i(\omega)> c\}\cap \{\omega: X_c(\omega)\in\Gamma,\quad \tau_i(\omega)\le t\}).\]
Since $\{\omega: \tau_i(\omega)\le c\}, \{\omega: \tau_i(\omega)> c\}$, $\{\omega: X_c(\omega)\in\Gamma\}\in \mathcal{F}_c\subset \mathcal{F}_t$ and $B\in \mathcal{F}_t$, $A\in \mathcal{F}_t$ as well.
\end{proof}



 Recall that $\mathcal{P}= \mathcal{P}(P) $ is the set of all martingale probability measures equivalent to $P$ and $\mathbb{E}_Q(Y)$ denotes expectation with respect to probability measure $Q$.

\begin{lemma}\label{martingaleDifference}
Let $Q \in \mathcal{P}(P)$ and
assume, for $i\ge 0$, $U_i$ are $\mathcal{F}_{\tau_i}$-measurable bounded functions, $|X_\tau|$ $Q$-integrable, and define
\[
Y_0\equiv 0\quad \mbox{and}\quad Y_n \equiv \sum^{n-1}_{i=0}U_i(X_{\tau_{i+1}}-X_{\tau_i})\quad \mbox{for}\quad n\ge 1.
\]
Then

($a$) $\{Y_n\} _{n\ge 0}$ is a martingale w.r.t. $\mathcal{F}_{\tau}= \{\mathcal{F}_{\tau_i}\}_{i\ge 0}$ and $Q$.


\noindent Assume $M_{\tau}$ is a $\mathcal{F}_{\tau}$-stopping time, and $M_{\tau}$ is bounded or $|Y_n|$ is bounded uniformly by an integrable function. For any $k\ge 0$,

($b$) $\mathbb{E}_Q[\sum_{i=k}^{M_{\tau}-1} U_i ~(X_{\tau_{i+1}} - X_{\tau_i})|\mathcal{F}_{\tau_k}] = 0$.
\end{lemma}
\begin{proof}
For ($a$), fix $n\ge 0$, then it holds that
\[\mathbb{E}_Q[Y_{n+1}|\mathcal{F}_{\tau_n}] = \mathbb{E}_Q[\sum_{i=0}^{n} U_i ~(X_{\tau_{i+1}} - X_{\tau_i})|\mathcal{F}_{\tau_n}]= \sum_{i=0}^{n-1} U_i ~(X_{\tau_{i+1}} - X_{\tau_i})=Y_n,\]
since for $0\le i \le n$, $U_i, X_{\tau_i}$ are $\mathcal{F}_{\tau_i}$-measurable, so $\mathbb{E}_Q[(X_{\tau_{n+1}} - X_{\tau_n})|\mathcal{F}_{\tau_n}]=0$.

For ($b$), first observe that $Z_n\equiv \sum_{i=k}^{n-1} U_i ~(X_{\tau_{i+1}} - X_{\tau_i}) =Y_n-Y_k$ for $n>k$, and $Z_n\equiv 0$ if $0\le n\le k$ is also a martingale w.r.t. $\mathcal{F}_{\tau}$.

For $m\ge1$, consider $\sigma_m:\Omega_0(\tau)\rightarrow \mathbb{N}$, defined as follows:
\[
\sigma_m(\omega)=\left\{
\begin{array}{ccc}
k & if & M_{\tau}(\omega)< k+1 \\
M_{\tau}(\omega) & if &  k+1 \le M_{\tau}(\omega) \le m \\
m & if & k+1 \le m < M_{\tau}(\omega).
\end{array}
\right .
\]
$\sigma_m$ is clearly bounded, and also a $\mathcal\mathcal{F}_{\tau}$-stopping time. It follows from
\[\{\sigma_m \le n\} = (\{k \le n\} \cap \{M_{\tau}< k+1\}) \cup (\{M_{\tau} \le n\} \cap \{M_{\tau}\ge k+1\} \cap \{M_{\tau}\le m\}) \cup (\{m \le n\} \cap \{M_{\tau} > m\}).\]
Since, if $n\le k$, the second and third sets in the union are empty, and the first one belongs to $\mathcal{F}_{\tau_n}$.

\noindent For $k+1\le n < m$, the third is empty, the first is in $\mathcal{F}_{\tau_n}$, and \[\{M_{\tau} \le n\} \cap \{M_{\tau}\ge k+1\} \cap \{M_{\tau}\le m\})=\{M_{\tau} \le n\} \cap \{M_{\tau}\le k+1\})\in \mathcal{F}_{\tau_n}.\]

\noindent On the other hand, if $k+1\le m \le n$ the first and third sets belongs clearly to $\mathcal{F}_{\tau_n}$, and \[\{M_{\tau} \le n\} \cap \{M_{\tau}\ge k+1\} \cap \{M_{\tau}\le m\}) = \{M_{\tau}\ge k+1\} \cap \{M_{\tau} \le m\}) \in \mathcal{F}_{\tau_m} \subset \mathcal{F}_{\tau_n}.\]




It follows, from \cite[Prop 1.83]{medvegyev}, using the stopping time $\sigma_0 \equiv k \le \sigma_m$, that
\[\mathbb{E}_Q[Z_{\sigma_m}|\mathcal{F}_{\tau_k}]=\mathbb{E}_Q[\sum_{i=k}^{\sigma_m-1} U_i ~(X_{\tau_{i+1}} - X_{\tau_i})|\mathcal{F}_{\tau_k}] = \mathbb{E}_Q[Y_{\sigma_m}-Y_k|\mathcal{F}_{\tau_k}]= 0.\]

Observe that if $M_{\tau}$ is bounded $Z_{M_{\tau}} =Z_{\sigma_m}$ for some $m\ge 1$. In general, $Z_{\sigma_m} \rightarrow Z_{M_{\tau}}$ pointwise, and since $|Z_{\sigma_m}|$ results bounded by an integrable function, we have
\[
\mathbb{E}_Q[\sum_{i=k}^{M_{\tau}-1} U_i ~(X_{\tau_{i+1}} - X_{\tau_i})|\mathcal{F}_{\tau_k}] = \lim_{m\rightarrow\infty} \mathbb{E}_Q[Z_{\sigma_m}|\mathcal{F}_{\tau_k}] = 0.
\]
\end{proof}

\begin{lemma}  \label{measurabilityOfOmega}
For a given $\omega \in \Omega_0(\tau)$ and $k\ge 0$, the set $\Omega_{\omega, k}(\tau)$, defined in Theorem \ref{veryUseful}, belongs to $\mathcal{F}_{\tau_k}$.
\end{lemma}
\begin{proof}
We have to show that $A=\Omega_{\omega, k}(\tau)\cap\{\tau_k\le t\}\in \mathcal{F}_{\tau_k}$, for all $t\ge 0$. Setting $C_{\omega, k}= C \cap [0, \tau_k(\omega)]$, observe that $\Omega_{\omega, k}(\tau)$ can be decomposed in the following way
\[
\Omega_{\omega, k}(\tau)=\left( \bigcap_{s\in C_{\omega, k}}\{\omega':X_s(\omega')=X_s(\omega)\}\right)\cap \left( \bigcap_{i=0}^k\{\omega':\tau_i(\omega')=\tau_i(\omega)\}\right).
\]
Note that if $t<\tau_k(\omega)$, then $\{\omega':\tau_k(\omega')=\tau_k(\omega)\}\cap \{\tau_k\le t\} = \emptyset$, thus $A=\emptyset\in\mathcal{F}_{\tau_k}$. Consequently it is enough to consider $\tau_k(\omega)\le t$.
Let $s\in C_{\omega, k}$, it follows that $s\le t$ and
\[
\{\omega':X_s(\omega')=X_s(\omega)\} = X_s^{-1}(\{X_s(\omega)\})\in \mathcal{F}_s\subset \mathcal{F}_t.
\]
Therefore $X_s^{-1}(\{X_s(\omega)\})\in \mathcal{F}_{\tau_k}$. On the other hand, since for $0\le i\le k$, $\tau_i\le \tau_k$ then $\tau_i$ are $\mathcal{F}_{\tau_k}$-measurables \cite{shiryaev}, which concludes that $\Omega_{\omega, k}(\tau)\in \mathcal{F}_{\tau_k}$.
\end{proof}

\noindent
{\bf Acknowledgments:}
S. Ferrando  would like to thank Zsolt Bihary and H.F\"{o}llmer
for stimulating discussions.



\end{document}